\documentclass[conference,compsoc]{IEEEtran}

\usepackage{amsmath}
\usepackage{amssymb}
\usepackage{amsthm}
\usepackage{amsfonts}
\usepackage{graphicx}
\usepackage{underscore}
\usepackage[english]{babel}
\usepackage{stmaryrd}
\usepackage{url}
\usepackage{algorithm}
\usepackage[noEnd=true,indLines=false]{algpseudocodex}
\usepackage{multicol}
\usepackage{changepage}
\usepackage{comment}
\usepackage[utf8]{inputenc}
\usepackage[english]{babel}
\usepackage{multirow}
\usepackage{tikz}
\usepackage{pgfplots}
\usetikzlibrary{patterns}
\usepackage{mathtools}
\usepackage[normalem]{ulem}
\usepackage{todonotes}
\setuptodonotes{inline}
\usepackage{pifont}
\usepackage{float}

\newtheorem{theorem}{Theorem}
\newtheorem{definition}{Definition}
\newtheorem{lemma}{Lemma}
\newtheorem*{remark}{Remark}

\usepackage{booktabs}
\usepackage{makecell}

\usepackage{threeparttable}
\usepackage{xcolor}
\usepackage[commandnameprefix=ifneeded,commentmarkup=uwave]{changes}
\usepackage[n,landau,notions,ff,mm]{cryptocode}
\usepackage{xspace}
\usepackage{placeins} 
\usepackage[framemethod=TikZ]{mdframed}
\usepackage{dashbox}
\usepackage{booktabs}
\usepackage{paralist}
\usepackage{hyperref}
\usepackage[noabbrev,capitalize]{cleveref}
\usepackage{makecell}
\usepackage{colortbl}
\usepackage{xcolor} 
\usetikzlibrary{shapes, positioning, arrows.meta}

\usepackage{pgf-umlsd}

\usepackage{amsthm}
\usepackage{xcolor}

\colorlet{shadecolor}{gray!20}

\newtheorem{insightenv}{Insight}

\newcommand{\insightbox}[1]{
  \noindent
  \fcolorbox{black}{shadecolor}{%
    \begin{minipage}{0.97\columnwidth}
      \begin{insightenv}
        #1
      \end{insightenv}
    \end{minipage}%
  }
}
\usepackage{tabularx, booktabs}

\usepackage{colortbl}
\usepackage{xcolor} 
\usepackage{tikz}
\usetikzlibrary{shapes, positioning, arrows.meta, calc}

\newcommand{\IMess}[4][0]{
  \stepcounter{seqlevel}
  \path
  (#2)+(0,-\theseqlevel*\unitfactor-0.7*\unitfactor) node (mess from) {};
  \addtocounter{seqlevel}{#1}
  \path
  (#4)+(0,-\theseqlevel*\unitfactor-0.7*\unitfactor) node (mess to) {};
  \path[->,>=angle 60] (mess from) -- (mess to) node[midway, above]
  {\footnotesize #3};
}

\newcommand{\Mess}[4][0]{
  \stepcounter{seqlevel}
  \path
  (#2)+(0,-\theseqlevel*\unitfactor-0.7*\unitfactor) node (mess from) {};
  \addtocounter{seqlevel}{#1}
  \path
  (#4)+(0,-\theseqlevel*\unitfactor-0.7*\unitfactor) node (mess to) {};
  \draw[->,>=angle 60] (mess from) -- (mess to) node[midway, above]
  {\footnotesize #3};
}

\def\eg{\textit{e.g.}\xspace}
\def\etal{\textit{et~al.}\xspace}
\def\ie{\textit{i.e.}\xspace}

\newcommand{\sysname}{\ensuremath{\mathtt{Lite\text{-}PoT}}\xspace} 

\newcommand{\defn}{\triangleq}

\newcommand{\pp}{\ensuremath{\mathsf{pp}}}
\newcommand{\gpone}{\mathbb{G}_{1}}
\newcommand{\gptwo}{\mathbb{G}_{2}}
\newcommand{\gptarget}{\mathbb{G}_T}
\newcommand{\genone}{I_{\gpone}}
\newcommand{\gentwo}{I_{\gptwo}}
\newcommand{\gentarget}{I_{\gptarget}}
\newcommand{\Zp}{\mathbb{Z}_p}
\newcommand{\isequal}{=_{?}}
\newcommand{\isnotequal}{{\neq}_?}
\newcommand{\histidx}{t}
\newcommand{\pk}{\mathsf{pk}}
\newcommand{\sk}{\mathsf{sk}}
\newcommand{\vk}{\mathsf{vk}}
\newcommand{\st}{\mathsf{st}}

\newcommand{\server}{\mathsf{S}}
\newcommand{\user}{\mathsf{U}}
\newcommand{\contract}{\mathcal{C}}
\newcommand{\presig}{\sigma_{\mathsf{prv}}}
\newcommand{\cursig}{\sigma_{\mathsf{cur}}}
\newcommand{\prevk}{\vk_{\mathsf{pre}}}
\newcommand{\curvk}{\vk_{\mathsf{cur}}}
\newcommand{\stkey}{\mathsf{st}_{\pk}}

\newcommand{\stser}{\mathsf{st}_{\mathsf{ser}}}
\newcommand{\tx}{\mathsf{tx}}
\newcommand{\txser}{\mathsf{tx}_{\contract}}
\newcommand{\txcontract}{\mathsf{tx}_{\contract}}

\newcommand{\stcontract}{\mathsf{st}_{\contract}}

\newcommand{\sample}{\leftarrow_{\$}}

\newcommand{\com}{\fun{Com}}
\newcommand{\compp}{\mathsf{compp}}
\newcommand{\arrcompp}{\mathtt{COMPP}}
\newcommand{\checkwellform}{\fun{CheckWellForm}}
\newcommand{\checkinclusion}{\fun{CheckInclusion}}

\newcommand{\inclusionsoundgame}{\mathbf{G}_{\mathsf{incl-snd}}}
\newcommand{\gamesdh}{\mathbf{G}_{(n, k)-\mathsf{sdh}}}
\newcommand{\advinclude}{\ensuremath{\mathcal{A}_{\mathsf{snd}}}}
\newcommand{\advb}{\ensuremath{\mathbf{B}}}
\newcommand{\funonupdate}{\fun{OnUpdate}}
\newcommand{\funoffupdate}{\fun{OffUpdate}}
\newcommand{\funcheckinclude}{\fun{CheckInclude}}
\newcommand{\funfraudproof}{\fun{Disprove}}
\newcommand{\funextwellform}{\checkwellform}

\newcommand{\idxill}{\ensuremath{i_{\mathsf{ill}}}}
\newcommand{\idxgp}{\ensuremath{i_{\mathsf{gp}}}}
\newcommand{\memprove}{\mathsf{MPrv}}

\newcommand{\memverify}{\mathsf{MVrf}}

\newcommand{\arrelem}{\mathcal{E}}
\newcommand{\pos}{\mathsf{PoP}}
\newcommand{\pop}{\mathsf{PoP}}
\newcommand{\coeff}[2]{c_{#1}[#2]}
\newcommand{\event}[1]{\ensuremath{\mathsf{E}_{#1}}}
\newcommand{\pr}[1]{\ensuremath{\mathsf{Pr[#1]}}}
\newcommand{\calg}{\ensuremath{\mathsf{C_{alg}}}}
\newcommand{\dalg}{\ensuremath{\mathsf{D_{alg}}}}
\newcommand{\ealg}{\ensuremath{\mathsf{E_{alg}}}}
\newcommand{\falg}{\ensuremath{\mathsf{F_{alg}}}}
\newcommand{\faalg}{\ensuremath{\mathsf{FA_{alg}}}}
\newcommand{\fbalg}{\ensuremath{\mathsf{FB_{alg}}}}
\newcommand{\jalg}{\ensuremath{\mathsf{J_{alg}}}}
\newcommand{\eqp}{\ensuremath{=}}
\newcommand{\numhashquery}{\ensuremath{q_{H}}}
\newcommand{\hashgp}{\ensuremath{H_{\gptwo}}}
\newcommand{\hashzp}{\ensuremath{H_{\Zp}}}
\newcommand{\checkone}{\ensuremath{\mathtt{KnowledgeCheck}}}
\newcommand{\checktwo}{\ensuremath{\mathtt{WellformCheck}}}
\newcommand{\checkthree}{\ensuremath{\mathtt{NonDegenCheck}}}

\newcommand{\todotemplate}[3]{%
    \mbox{}
    \marginpar{%
        \colorbox{#2!80!black}{\textcolor{white}{#1}}%
        \vspace*{-22pt}
    }%
    \textcolor{#2}{{#3}}%
}

\newcommand{\lucien}[1]{
\todotemplate{lucien}{purple}{#1}
}

\newcommand{\panos}[1]{
\todotemplate{panos}{red}{#1}
}

\newcommand{\pedro}[1]{
\todotemplate{pedro}{teal}{#1}
}

\newcommand{\pparagraph}[1]{\vspace{0.5em} \noindent \textbf{#1.}}
\newcommand{\zksnark}{zk-SNARK\xspace}
\newcommand{\var}[1]{\ensuremath{\mathsf{#1}}}

\newcommand{\fun}[1]{\ensuremath{\mathtt{#1}}\xspace}

\newcommand{\poly}{\mathsf{poly}}

\newcommand{\adv}{\ensuremath{\mathcal{A}}}

\newcommand{\rt}{\ensuremath{\mathsf{rt}\xspace}}

\newcommand*\circled[1]{\tikz[baseline=(char.base)]{
            \node[shape=circle,fill,inner sep=0.8pt] 
            (char) {\textcolor{white}{#1}};}}

\newcommand{\degree}{\ensuremath{d}}
\newcommand{\ppArray}[2]{\ensuremath{#1[#2]}}

\begin{document}

\title{\sysname: {P}ractical {P}owers-{o}f-{T}au Setup Ceremony}
\makeatletter
\newcommand{\linebreakand}{%
  \end{@IEEEauthorhalign}
  \hfill\mbox{}\par
  \mbox{}\hfill\begin{@IEEEauthorhalign}
}
\makeatother

\author{\IEEEauthorblockN{Lucien K. L. Ng}
\IEEEauthorblockA{Georgia Institute of Technology}
\and
\IEEEauthorblockN{Pedro Moreno-Sanchez}
\IEEEauthorblockA{IMDEA Software Institute\\
Visa Research}
\and
\IEEEauthorblockN{Mohsen Minaei}
\IEEEauthorblockA{Visa Research}
\and
\IEEEauthorblockN{Panagiotis Chatzigiannis}
\IEEEauthorblockA{Visa Research}
\linebreakand
\IEEEauthorblockN{Adithya Bhat}
\IEEEauthorblockA{Visa Research}
\and
\IEEEauthorblockN{Duc V. Le}
\IEEEauthorblockA{Visa Research}}

\maketitle
\begin{abstract}
    Zero-Knowledge Succinct Non-Interactive Argument of Knowledge (\zksnark)
    schemes have gained 
    significant adoption in privacy-preserving applications, decentralized systems (e.g., blockchain), and verifiable computation  
    due to their efficiency. 
    However, the most efficient zk-SNARKs often rely on a one-time trusted setup to generate a public parameter, often known as the ``Powers of Tau" (PoT) string. 
    The leakage of the secret parameter, $\tau$, in the string would allow attackers to
    generate false proofs, compromising the soundness of all \zksnark systems built on it. 

    Prior proposals for decentralized setup ceremonies have utilized blockchain-based smart contracts to allow any party to contribute randomness to $\tau$ while also preventing censorship of contributions. For a PoT string of \degree-degree generated by the randomness of $m$ contributors, these solutions required a total of $O(m \degree)$ on-chain operations (i.e., in terms of both storage and cryptographic operations). These operations primarily consisted of costly group operations, particularly scalar multiplication on pairing curves, which discouraged participation and limited the impact of decentralization.

    In this work, we present \sysname, which includes two key protocols designed to reduce participation costs:
    \emph{(i)} a fraud-proof protocol to reduce the number of expensive on-chain cryptographic group operations to $O(1)$ per contributor.
    Our experimental results show that
    (with one transaction per update)
    our protocol enables decentralized ceremonies for PoT strings up to a $2^{15}$-degree, an ${\approx}16\times$ improvement over existing on-chain solutions;
    \emph{(ii)} a proof aggregation technique that batches $m$ randomness contributions into one on-chain update with only $O(d)$ on-chain operations, independent of $m$.
     This significantly reduces the monetary cost of on-chain updates by $m$-fold via amortization. 




\end{abstract}


\section{Introduction}
\label{sec:intro}



Many cryptographic protocols depend on an initial setup to generate \emph{public parameters}. This is a critical step in preserving their security properties. 
This setup also creates an unwanted \emph{trapdoor} (also called a backdoor or toxic waste in the literature). If anyone learns this trapdoor, it would undermine the security properties of the system using those parameters. 
The \emph{powers-of-tau} (PoT) is a prominent example of such a setup. It is used in KZG polynomial commitments~\cite{aniket-polycommit-2010}, Verkle trees~\cite{verkle-tree}, fast proofs of multi-scalar multiplication (MSM)~\cite{msm-proofs}, the Danksharding proposal~\cite{danksharding} and  zk-SNARKs~\cite{sonic-2019, plonk-eprint-2019}. 
In fact, the increasing use of zk-SNARKS in the blockchain domain make powers-of-tau setup an indispensable tool within a variety of blockchain applications, including zk-rollups~\cite{rollups}, zkBridges~\cite{zkbridge} and privacy-preserving payment systems such as ZeroCash~\cite{zerocash}.  

In PoT, the public parameters ($\pp$) are defined as: 

{\footnotesize
\vspace{-0.3cm}
\begin{align*}
  \pp
  = (&\tau \genone, \tau^2 \genone, \ldots, \tau^n \genone; 
   \tau \gentwo, \tau^2 \gentwo, \ldots, \tau^k \gentwo )
\in \mathbb{G}_1^n \times \mathbb{G}_2^k
\end{align*}
}

Here, $\genone$ and $\gentwo$ are the generators of elliptic curve groups $\mathbb{G}_1$ and $\mathbb{G}_2$, respectively.  
We hereby denote the value $\degree = n+k$ as the \emph{number of powers-of-tau}, where 
the value $\tau$ is the trapdoor. 
{Supporting a large number of powers-of-tau is crucial for enhancing the expressiveness of an application using $\pp$. For example, in zk-SNARKs, the number of powers-of-tau dictates how complex the statement of the proof can be.}
To ensure the security of an application using $\pp$, it is crucial that $\tau$ is generated randomly and remains unknown to anyone. 
As a straightforward solution, one could rely on a trusted party to create $\pp$ in a so-called \emph{trusted setup}. 
This party selects $\tau$ at random, generates $\pp$, and deletes $\tau$ (and the randomness used to create it) as soon as the generation of $\pp$ is completed. 
However, trusted setups are not desirable in practice for decentralized applications as finding a universally accepted trusted party can be challenging.

Given this problem, there is a clear motivation for a distributed protocol for {efficiently} generating powers-of-tau parameters, which allows multiple users to participate.
To realize such a protocol, a key observation about the powers-of-tau $\pp$ is that it is \emph{re-randomizable}. 
Given $\pp$ and without knowing $\tau$, one can produce a new string $\pp'$ by choosing a new random value {$r$}  
and multiplying each component of $\pp$ by an appropriate power of 
{$r$}, more concretely: 

\vspace{-0.3cm}
{\small
\begin{align*}
  \pp'  
  = (&r \cdot \tau \genone, \ldots, r^n \cdot \tau^n \genone;
  r\cdot \tau \gentwo, \ldots, r^k \cdot \tau^k \gentwo )
\end{align*}
}

Here, the new trapdoor becomes $r \cdot \tau$, which is secure as long as \emph{either} $\tau$ or $r$ are unknown and \emph{neither} is zero. 
This re-randomizability property of powers-of-tau permits a relatively simple serial MPC protocol (also called \emph{ceremony}): the $i$-th user in turn randomizes the public parameters with their own chosen value {$r_i$}.  
As long as each participant correctly re-randomizes the $\pp$ from the previous user and at least one of the participants destroys their value {$r_i$} (and the randomness used to draw it), the cumulatively constructed public parameter will be secure. 

However, realizing such a ceremony in practice presents several challenges: 
\emph{(i)} all participants should agree on the final value of $\pp$ (\emph{consensus}); 
\emph{(ii)} each participant should only be able to re-randomize the current $\pp$ and not simply replace it with one for which the trapdoor is known to them (\emph{validity}); 
\emph{(iii)} the final $\pp$ should be available for all participants, as well as the history of prior versions for auditability (\emph{data availability}); 
\emph{(iv)} it should not be possible to prevent any participant from contributing to the ceremony (\emph{censorship resistance}). 
\subsection{State-of-the-art and Goals}

Nikolaenko et al.~\cite{boneh-pot-2024} recently proposed the \emph{first truly decentralized and permissionless setup ceremony} for powers-of-tau. 
They observed that one can leverage a smart contract in a blockchain with expressive smart contracts (like Ethereum)  to coordinate the contributions of different participants of the ceremony. In a nutshell, upon receiving a contribution from a participant, the smart contract 
first checks its validity and 
then stores it as the last valid $\pp$. 
In doing so, this approach seamlessly leverages existing blockchain platforms to achieve part of the desired properties. 
In particular:  
\emph{(i)} consensus is inherently provided by the blockchain; 
\emph{(ii-iii)} the validity of each contribution as well as data availability are ensured by the smart contract itself; 
\emph{(iv)} the permissionless nature of existing blockchains makes them naturally resistant to censorship.

However, the validity check required by the smart contract has a cost in terms of gas (i.e., the fee that the user needs to pay in Ethereum to execute the contract) that is linear in \degree. 
{More significantly, this validity check involves costly elliptic curve group operations, making it expensive for users to contribute}.
For instance, even for contributing a {$\pp$ with a} small number of \degree, {e.g. $2^{10}$ elements,} a user's contribution will require $11,500,000$ gas (i.e., \$315 as reported in~\cite{boneh-pot-2024}).  
Note that limiting the number of powers-of-tau in $\pp$ restricts the applications for which the generated $\pp$ can be used. 
For instance, in zk-SNARKS, the size of the statement that can be proven depends on the number of elements within powers-of-tau (i.e., fewer {values of \degree} can only work with small statements).



The goal of this work is to answer the question: 
\begin{center}
    \emph{Can we devise a decentralized setup ceremony that permits a larger number of powers-of-tau while maintaining consensus, validity, data availability, and censorship resistance?}
\end{center}

\subsection{Our Contributions}

In this work, we positively answer this question. For this, our contributions are as follows.

First, we introduce \sysname, a fully decentralized and permissionless setup ceremony.  
Unlike the approach in~\cite{boneh-pot-2024}, where each $\pp$ update is proactively checked for two conditions before acceptance -- \emph{(i)} the contributor knows the value $r$ used in the contribution, and \emph{(ii)} the proposed $\pp$ is well-formed -- \sysname's key innovation is a fraud-proof protocol that shifts the paradigm for creating $\pp$. In \sysname, each proposed $\pp$ update is only checked for condition \emph{(i)} before acceptance, while condition \emph{(ii)} is not proactively verified. If a user \emph{a posteriori} identifies a $\pp$ update as invalid under condition (ii), they can submit a succinct fraud-proof. If the contract verifies this fraud-proof, it rejects the invalid $\pp$ update and reverts to the previous valid $\pp$ update.

{This paradigm shift allows us to design the fraud-proof protocol in a way that significantly reduces gas consumption compared to~\cite{boneh-pot-2024}. This reduction enables setup ceremonies with a larger number of powers-of-tau while maintaining key properties such as consensus, validity, data availability, and censorship resistance.}

Second, to further enhance the efficiency of \sysname, we introduce an aggregation technique for combining multiple contributions while keeping the on-chain cost constant. In this setting, users do not submit their $\pp$ and proof for fact 
\emph{(i)} directly to the contract. Instead, they send their contributions to an \emph{untrusted} operator, who collects and forwards them to the contract. A key innovation here is our proof aggregation protocol, which allows the operator to batch $m$ tuples of the form ($\pp_j$, $\pi_j$)—where $\pp_j$ is the contributed $\pp$ and $\pi_j$ is the proof for fact \emph{(i)}—into a single $\pp^*$ and a constant-size proof $\pi^*$ 
consisting of only $3$ group elements.
When the operator submits $(\pp^*, \pi^*)$ to the contract, the contract verifies $\pi^*$ and accepts $\pp^*$ as the latest contribution. As with \sysname, $\pp^*$ is not proactively checked for validity under fact \emph{(ii)}, so users can still challenge its validity later using the fraud-proof designed in \sysname.
This aggregatable contribution~ reduces the on-chain cost from $O(m\cdot \degree)$ -- which would be required to store and process $m$ individual contributions -- to $O(\degree)$. This makes the on-chain gas cost  
independent to 
the number of contributions 
batched off-chain. 
Looking ahead, our construction employs bilinear pairing following the BLS signature scheme~\cite{asiacrypt/BonehLS01}, which is known to be vulnerable to rogue-key attacks if used in a multi-party setting. To address this vulnerability, we implement multiple countermeasures in our proof aggregation protocol and prove the security of our scheme in the algebraic group model (AGM).

Third, we implemented both the on-chain fraud-proof mechanism and the off-chain aggregation techniques in \sysname. Our performance evaluation shows that, even without aggregation, \sysname can reduce on-chain gas consumption by $16$-fold compared to the state-of-the-art solution~\cite{boneh-pot-2024}. More concretely, the efficiency boost of \sysname allows creating a $\pp$ with a number of powers-of-tau of $2^{15}$, thereby enabling for the first time its usage in applications like ProtoDanksharding~\cite{eip4844}, that requires $\pp$ with a number of powers-of-tau between $2^{12}$ and $2^{15}$. 
As a side benefit of our off-chain collaboration paradigm,
our evaluation demonstrates an order-of-magnitude speedup because of the optimized $\pp$ verification.
With our aggregation technique, the contributors can move most of the computation off-chain without worrying about increasing the on-chain cost.
The $\pp$ verification after each sequential update, which is one of the bottlenecks in the ceremony, 
can be shifted off-chain.
Contributors can now adopt optimization, \eg, parallelization, to accelerate the off-chain computation.
It greatly reduces the verification time compared to the prior art~\cite{boneh-pot-2024}, which entirely relies on the Ethereum VM
-- 
an environment that lacks support for parallelism and is constrained in its ability to perform high-performance computations.

\section{Preliminaries}
\pparagraph{Notation}
We denote by $(\dots)$ a list.
$||$ is the concatenation operator. Note that an object can be appended into a set or list via concatenation,
\eg, $(a, b) || c = (a, b, c)$.
When a list has named variables, say, $A = (x, y)$, we can access the variables by their names, \eg, $A[x] = x$ or $A.x$.
The $x$ here is a symbol instead of the value of $x$.
We use $*$ as a wildcard and ``$;$'' as a separator in a 2-dimension list.
We use $\defn$ when defining a function.
We denote by $[m]$ the list $(1, \ldots, m)$, by $\bot$ an null element,
by 
$=_{?}$ the equality check operator that outputs $1$ if $A = B$ and $0$ otherwise.

We write $x \sample X$ to denote sampling an element $x$ from a set $X$ uniformly and independently at random.
$\gpone, \gptwo, \gptarget$ are elliptic curve groups.
We use addition notation for group operations, \ie, group elements are added instead of multiplied, and a scalar multiplies a group element instead of using exponents.
$\hashgp: \{0, 1\}^{*} \to \gptwo$ and $\hashzp: \{0, 1\}^{*} \rightarrow \Zp$ are cryptographic hash function.
We use $\poly(\lambda)$ to denote a polynomial function in the security parameter $\lambda$.

\begin{table}[t]
\center
\caption{Notation} \label{tb:notation}
\small
\resizebox{\columnwidth}{!}{ 
\begin{tabular}{r|l}
  \toprule
  $\pp$ & Public Parameter in the form of Powers-of-Tau \\
  $r$ & Random Value Multiplied to $\pp$ for Update \\
  $n, k$ & $(\gpone, \gptwo)$ Degrees of $\pp$, respectively  \\
  $d$ & $n + k$, the total degree of $\pp$  \\
  $m$ & Number of Contributors in a Batch \\
  $P_1, \ldots, P_n$ & Elements in $\pp$'s $\gpone$ Array \\
  $Q_1, \ldots, Q_k$ & Elements in $\pp$'s $\gptwo$ Array \\
  $\user$ & Contributor  \\
  $\server$ & Operator  \\
  $\contract$ & Smart Contract Holding the Setup Ceremony \\
  $\genone, \gentwo, \gentarget$ & Generator in $\gpone$, $\gptwo$, and $\gptarget$, respectively \\
  $\lambda$ & Security Parameter \\
  \bottomrule
  \hline
\end{tabular}
 }
\end{table}

\pparagraph{Threat Model and Entities}
Our protocols involve three types of entities:
a \emph{smart contract} $\contract$,
\emph{operator(s)} $\server$,
and \emph{contributors} $\user$.
The smart contract $\contract$ acts as a trusted party, executing its predefined code reliably and maintaining data integrity. 
Contributors mainly interact with an operator to contribute their randomness to the $\pp$.
Operators organize the contributions, update the $\pp$, and eventually upload the final $\pp$ and its proofs to $\contract$.
Note that operators do not hold any private state; they simply route messages between contributors and $\contract$.
Both contributors and operators can access the transaction (execution) history, view the state of $\contract$, and submit transactions to $\contract$ for execution.

We assume a permissionless blockchain, allowing anyone to participate as an operator or contributor.
Additionally, we consider a computationally bounded adversary who can corrupt any number of operators and all but one contributor.





\subsection{Cryptographic Tools}

\pparagraph{Vector Commitment} We use vector commitment as one of the building blocks for our protocol.
\begin{definition}[Vector Commitment~\cite{boneh2020graduate}]
A vector commitment scheme consists of four algorithms as follows:
{
 \begin{itemize}
   \item $\fun{Init}(1^{\lambda}) \rightarrow \mathsf{par}$
   takes as input the security parameter $1^{\lambda}$
   and outputs the system parameter $\mathsf{par}$, which is implicitly taken by other algorithms as input.
   \item $\com(A) \rightarrow \rt$
   takes as input an array $A$ and outputs a vector commitment $\rt$.
   \item $\memprove(A, \mathsf{pos}) \rightarrow \mathsf{prf}$
   takes an array $A$ and an index $\mathsf{pos}$ and outputs
   a positional membership proof $\mathsf{prf}$.
   %
   \item $\memverify(\rt, \mathsf{pos}, \mathsf{prf}, e) \rightarrow \{0, 1\}$
   takes as input a vector commitment $\rt$, a position $\mathsf{pos}$, a membership proof, and an element $e$ and outputs a bit $0/1$.

  \end{itemize}
}
\end{definition}
A commitment scheme should satisfy correctness, positional binding, and succinctness.
For the formal
definitions of these properties, we refer readers to \cite{boneh2020graduate}.

Our protocols expect the vector commitment scheme to have $O(1)$-size commitments and $O(\log n)$-size membership proofs and $O(n)$-time verification for a $n$-size committed vector. It can be easily achieved by using Merkle Tree to realize the vector commitment scheme.

\pparagraph{Bilinear Pairing}
A bilinear pairing is a map:
$e: \gpone \times \gptwo \rightarrow \gptarget$
where $\gpone$, $\gptwo$, and $\gptarget$ are groups, typically elliptic curve groups, with prime order $p$ and the following properties:

\begin{compactenum}
    \item \textit{Bilinearity}: For all \( U \in \gpone \), \( V \in \gptwo \), and \( a, b \in \Zp \),
    $$e(a \cdot U, b \cdot V) = ab \cdot e(U, V)$$

    \item \textit{Non-degeneracy:} 
        $e(\genone, \gentwo) \neq \gentarget$

\end{compactenum}

\pparagraph{$(n,k)$-Strong Diffie-Hellman Assumption}
Following is the $(n,k)$-SDH assumption defined by a security game.
\begin{definition}[$\gamesdh^{\adv}$]
  The security game for $(n,k)$-Strong Diffie-Hellman assumption is defined as following:\\
  \begin{compactenum}
    \item The challenger samples $z \sample \Zp$ and sends to $\adv$
    {\footnotesize
    \begin{align*}
    (&\genone, z \genone, \ldots, {z^n} \genone; \gentwo, z \cdot \gentwo, \ldots, z^{k} \gentwo) \in \gpone^{n+1} \times \gptwo^{k+1}
    \end{align*}
    }
    \item $\adv$ returns $(c, U) \in \Zp \times \gpone$.
    \item The challenger returns $c \isnotequal -z \wedge U \isequal \genone^{\frac{1}{c + z}}$.
  \end{compactenum}
\end{definition}

\begin{definition}[$(n,k)$-strong Diffie-Hellman Assumption]\label{def:q-sdh}
  For any $n, k \in \poly(\lambda)$,
  any PTT adversary $\adv$ can win $\gamesdh^{\adv}$ with the neligible probability.
\end{definition}

\pparagraph{BLS Signature Scheme} 
We use BLS as an alternative approach for proof-of-possession in our protocol.
\begin{definition}[BLS Signature~\cite{asiacrypt/BonehLS01}] BLS signature scheme consists of the following algorithms: 
\begin{compactitem}
    \item $\mathsf{KeyGen}(1^{\lambda}) \rightarrow (\sk \sample \Zp^{*}, \pk = \sk \cdot \gentwo)$.
    \item $\mathsf{Sign}(m) \rightarrow \sigma = \sk \cdot H_{\gpone}(m)$
    \item $\mathsf{Verify}(\pk, m, \sigma) \rightarrow e(H_{\gpone}(m), \pk) =_{?} e(\sigma, \gentwo)$
\end{compactitem}
\end{definition}
This signature scheme is existentially unforgeable under adaptive chosen message attacks in the random model and under the computation Diffie-Hellman assumption,
which a weaker assumption than the $(n, k)$-SDH assumption.

BLS signature scheme can be easily extended for signature aggregation.
Suppose multiple signers with $\pk_{1}, \ldots, \pk_{m}$ have their signatures $\sigma_{1}, \ldots, \sigma_{m}$ on the same message.
They can combine their signatures into one: $\sigma^{*} = \sum_{i \in [m]} \sigma_{i}$ and verified by a verification key $\vk = \sum_{i \in [m]} \pk_{i}$.

However, this key aggregation scheme suffers from \emph{rogue key attacks}.
An adversary can forge an aggregate signature on an arbitrary message that has never been signed by the victim.
We discuss their mitigations in \Cref{sec:related-work}.


\pparagraph{Algebraic Group Model (AGM)~\cite{crypto/FuchsbauerKL18}} AGM is an idealized model of group operations for constructing security proofs. It white-boxes the group operations of a PPT adversary $\adv$.
Whenever $\adv$ outputs a group element $Y$, it also outputs an array of coefficients $(a_{1}, \ldots, a_{l})$ such that $U = \sum_{i=1}^l a_{i} X_{i}$,
where $\{X_{i}\}_{i \in [l]}$ are all the group elements $\adv$ has received.
Basically, it restricts $\adv$ to only output group elements as a linear combination of prior seen group elements.
AGM is used to prove the tight security reduction for the (original) BLS signature scheme~\cite{crypto/FuchsbauerKL18} and will serve as an important tool to prove the security of our proof aggregation scheme.

\subsection{Prior Decentralized PoT Ceremony Scheme}
\label{sec:boneh-pot}
In this part, we recall how the decentralized setup works as described in \cite{boneh-pot-2024}. We then present several insights on potential improvements to their constructions that will be later achieved by our protocol.



\pparagraph{Initialization} In~\cite{boneh-pot-2024},
there is a smart contract initialized with a public parameter $\pp_0$:
\begin{align*}
    \pp_0 &= (\genone, \ldots, \genone; \gentwo, \ldots, \gentwo)
\end{align*}
This initialization step can be thought of as the first contribution with $\tau = 1$.

\pparagraph{Power-of-Tau Update}
After the initialization step, the $i$-{th} contributor can update the current parameter, $\pp = (P_1, \ldots, P_n; Q_1, \ldots, Q_k)$ by submitting:
\begin{align*}
    \pp' &= (P'_1, \ldots, P'_n; Q'_1, \ldots, Q'_k)
        \\&=  (rP_1, \ldots, r^n P_n; r Q_1, \ldots, r^k Q_k)
\end{align*}
to the contract along with a cryptographic proof $\pi$.
The smart contract accepts $(\pp', \pi)$ and update $\pp \leftarrow \pp'$ if and only if the pair $(\pp', \pi)$ passes the following three checks:

      \pparagraph{(Check \#1) $\checkone$}
      This check binds the updated string with the previous one. The contributor uses $\pi$ to demonstrate, in zero-knowledge, the knowledge of $r \in \Zp^{*}$ such that
      $P'_{1} = r \cdot P_{1}$,

  This procedure safeguards against any \emph{malicious} attempt to replace the previous string with a new one. 
  For this check, the construction in~\cite{boneh-pot-2024} uses a Fiat-Shamir version of Schnorr's $\Sigma$-protocol for its on-chain efficiency.  
  
    \pparagraph{(Check \#2) $\checktwo$}
    This check  ensures that $\pp'$ is a well-formed powers-of-tau string,
        \ie, there exists $\tau \in \Zp^{*}$ such that
       $\pp_i = (\tau \genone, \tau^{2} \genone, \ldots, \tau^{n} \genone; \tau \gentwo, \ldots, \tau^{k} \gentwo)$. To perform $\checktwo$, the smart contract executes the following on-chain computation:
       
        {\scriptsize
        \begin{align*}
        &e\left(\sum_{i=1}^{n} \rho_1^{i-1} P_i, \ \gentwo + \sum_{\ell=1}^{k-1} \rho_2^{\ell} Q_{\ell}\right) \stackrel{?}{=} e\left(\genone + \sum_{i=2}^{n-1} \rho_1^{i} P_i, \ \sum_{\ell=1}^{k} \rho_2^{\ell-1} Q_{\ell}\right)
        \end{align*}
        }

    \pparagraph{(Check \#3) $\checkthree$}
    $\pp_i$ is not degenerated, \ie, $r \neq 0$. This is done by checking $P'_1 \neq \genone$.

\subsubsection{Comparison With Our Work}
Looking ahead, we make two main observations that differentiate this state-of-the-art work from our approach (c.f. \cref{sec:tech-overview}). 

\vspace{1em}
\insightbox{
  While Schnorr's $\Sigma$-protocol used in~\cite{boneh-pot-2024} is efficient for on-chain applications, it lacks the aggregation properties found in other BLS-style proof-of-possession. As such, it might not be the best choice to combine multiple contributions into one.
}
  
We will discuss in \Cref{sec:agg-schnorr} that it is not immediate to extend \cite{boneh-pot-2024} for aggregation.

  \vspace{1em}
  \insightbox{
      The check $\checktwo$ incurs substantial on-chain computation costs as it requires $O(\degree)$ elliptic curve multiplications.
      To avoid such expensive operations, we can store the $\pp$ in a more efficient way and employ a considerably more efficient $O(\log(\degree))$ fraud-proof check to handle invalid $\pp$.
    }

In light of these insights, this work introduces two novel techniques: an \emph{on-chain fraud-proof mechanism} (\cref{subsec:fraud-proof-overview}) and \emph{off-chain aggregatable contributions} (\cref{sub:agg-contribution}), reducing the cost of PoT updates.




\section{Technical Overview}
\label{sec:tech-overview}

The goal of this work is to design a decentralized setup ceremony that ensures key properties such as consensus, validity, data availability, and censorship resistance while supporting more powers-of-tau than existing solutions.

To achieve this, we present two key contributions depicted in~\cref{fig:protocol-overview}: \circled{1} a fraud-proof mechanism; and \circled{2} an aggregatable contribution mechanism.

\begin{figure*}[t]
  \centering
  \includegraphics[width=0.99\linewidth]{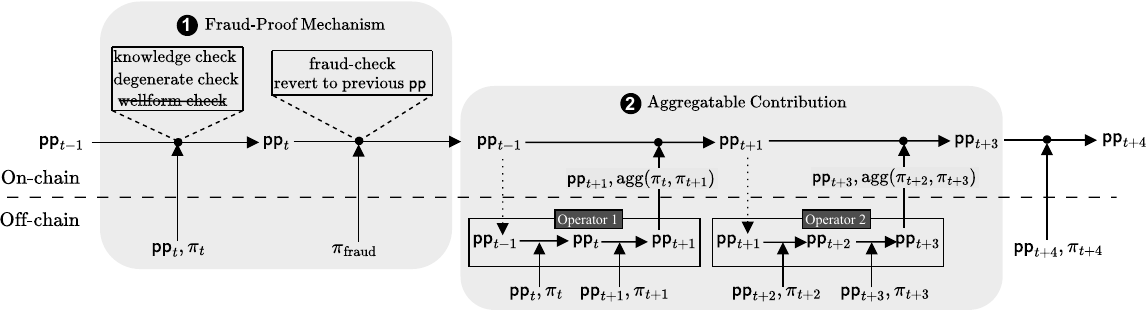}
  
  \caption{\protect\circled{1} provides an overview of our Fraud-Proof mechanism. For an update, the contract ($\contract$) only runs the inexpensive verification for $\checkone$ and $\checkthree$. Unlike in~\cite{boneh-pot-2024}, $\contract$ skips the costly verification $\checktwo$. Only when $\pp_t$ is ill-formed, a fraud-proof $\pi_{\mathtt{fraud}}$ will be submitted, and $\contract$ can verify $\pi_{\mathtt{fraud}}$ with $O(\log\degree)$ gas 
  and then revert back to $\pp_{t-1}$. \protect\circled{2} gives an overview of our Aggregatable Contribution Scheme. Here, an operator receives contributions from multiple contributors, verifies them off-chain, aggregates the proofs, and submits the latest $\pp$ with the constant-size proof. Another operator can then take over and aggregate subsequent contributions.}
\label{fig:protocol-overview}
\end{figure*}
\subsection{Fraud-Proof Mechanism}
\label{subsec:fraud-proof-overview}
\label{subsec:fraud-overview}

Our first key observation for reducing 
the ceremony cost
(c.f. \circled{1} in~\cref{fig:protocol-overview}) involves adopting an optimistic approach. 
Unlike~\cite{boneh-pot-2024}, where the contract explicitly verifies each update of $\pp$, our contract assumes each update is well-formed by default. 
If anyone successfully challenges an ill-formed $\pp$ using a fraud-proof, the contract reverts the state to the last valid $\pp$. 
Below, we detail this optimistic paradigm.


\pparagraph{Optimistic Paradigm}
In this paradigm, a user contributes to the power-of-tau generation with a tuple $(\pp_t, \pi_t)$, where $\pp_t$ is the user's contribution and $\pi_t$ is the zero-knowledge proof required for $\checkone$ in~\cite{boneh-pot-2024} (see~\cref{sec:boneh-pot}). 
The contract accepts the update $\pp_t$ only if $\checkone$ and $\checkthree$ are satisfied. 
Unlike~\cite{boneh-pot-2024}, $\checktwo$ is not proactively verified by the contract, allowing for the potential submission of an invalid contribution.


If any user detects an invalid 
contribution, $\pp_t$, they can submit a fraud proof, $\pi_\textit{fraud}$, to the contract. Upon a successful challenge, the contract rewinds to the state immediately preceding the invalid $\pp$, i.e., $\pp_{t-1}$, nullifying all subsequent updates. This approach incentivizes contributors to validate prior updates and challenge any invalid ones, ensuring the integrity of the process.

In the best-case scenario, where all the users are honest, this optimistic paradigm is more efficient than~\cite{boneh-pot-2024}, as $\checktwo$ is bypassed after each contribution.
However, in the worst-case scenario, if an invalid $\pp_t$ is identified, the contract must perform the expensive $\checktwo$ to validate the challenged contribution. 
While $\checkone$ and $\checkthree$ are efficient, they do not fully guarantee the well-formedness of $\pp$. The $\checktwo$ serves as a critical safeguard when discrepancies arise.

Next, we overview how to overcome this challenge.

\pparagraph{Fraud Proof Implementation}
To mitigate the potentially high costs associated with verifying a fraud proof, we introduce an optimized fraud-proof mechanism. This mechanism reduces the verification process to simply reading four group elements and performing two pairing equality checks, independent of \degree. Additionally, it requires only $O(\log \degree)$ hashing operations, which are significantly more gas-efficient compared to the pairing equality checks\footnote{Each pairing operation consumes $34,000$ gas compared to $30$ gas for a hashing in EVM-based chains.}.


The first pairing equality check is performed as follows:
\begin{align*}
    e(P_1, \gentwo) \isequal e(\genone, Q_1)
\end{align*}
where $P_1$ and $Q_1$ are extracted from $\pp = (P_1, \ldots, P_n; Q_1, \ldots, Q_k)$.
If this equality does not hold, the contract accepts the fraud-proof, deeming the challenge successful.

However, if the equality holds, we can safely assume that $P_1 = \tau\genone$ and $Q_1 = \tau \gentwo$ for some $\tau \in \Zp$.
Given this assumption, if $\pp$ is ill-formed but still passes the above equality test, it must be structured as follows:
\begin{align*}
    \pp = (\tau\genone, \ldots, \tau^{i-1}\genone, \textcolor{red}{\delta} \genone, \ldots; \quad \tau \gentwo, \ldots )
\end{align*}
where the $i$-th item in the array for $\gpone$ is incorrect, i.e., $\delta \neq \tau^i$ for an index $i > 1$, and it is the first incorrect item in $\gpone$. This implies that $P_{i-1} = \tau^{i-1}\genone$. Without loss of generality, we assume the error occurs in $\gpone$, though our approach also applies to errors in the array for $\gptwo$.

To prove the $i$-th item in $\gpone$ is incorrect, the prover submits the index $i$ to the contract, which then verifies the fraud-proof by computing:
\begin{align}\label{eqt:ill-pairing-check}
    e(P_{i-1}, Q_1) \isequal e(P_i, \gentwo)
\end{align}

Given our assumption that $P_i \neq \tau^{i}\genone$, the above equality must not hold, because the left-hand side (LHS) is:
\begin{align*}
e(P_{i-1}, Q_1) =
e(\tau^{i-1}\genone, \tau \gentwo)
=\tau^{i} \cdot e(\genone, \gentwo)
\end{align*}
while the right-hand side (RHS) is:
\begin{align*}
    e(P_i, \gentwo)
    = e(r\genone, \gentwo) = r \cdot e(\genone, \gentwo)
\end{align*}
Thus, LHS $\neq$ RHS because $r \neq \tau^{i}$.

While there may be other incorrect items in $\pp$, it is sufficient to prove $\pp$ is ill-formed by demonstrating that just one item is incorrect.

While the approach so far reduces the amount of computation required by the contract to verify the correctness of a fraud-proof, the verification still requires the contract to access (and thus store) the complete $\pp_t$ contribution, which is expensive when the number of powers-of-tau increases. 
Next, we overview our approach to overcome this challenge.

\pparagraph{Storing $\pp$ as Vector Commitments}
To reduce the storage overhead to verify a fraud proof, 
the $\contract$ only stores $\compp$, which is the Merkle root the submitted $\pp$\footnote{Although $\contract$ no longer stores $\pp$ in its state, $\server$ and $\user$ can still fetch the historic $\pp$ from blockchain transaction history
\eg, for fraud proofs.
}.
When a challenge is raised against one contribution $\pp$, the challenger also provide the membership proofs of the ill-formed elements in $\compp$.
As shown in~\cref{tb:fraud-cost}, computing $\compp$ requires $O(\degree)$ calldata and hashing, which is an order of magnitude cheaper than storing the $\pp$.
To handle a fraud-proof challenge, $\contract$ can verify the membership proof with an insignificant gas cost, requiring only $O(\log\degree)$ calldata and hashing operations.


\pparagraph{Practical Deployment and Cost} 
To deploy the aforementioned fraud-proof mechanism in practice, one should consider several aspects. 
First, a naive approach to the challenge mechanism might involve a dedicated fixed period for each $\pp$ update, during which no further updates occur.
However, this waiting period is unnecessary, as we discuss in~\cref{sec:discussion}. 
Second, a user that detects an invalid contribution to $\pp$ needs to pinpoint the ill-formed element within $\pp$ to create the (succinct) fraud-proof. As we describe in~\cref{sec:discussion}, there exists a mechanism for the user to create such a proof with only a small performance cost.

In summary, we have designed a fraud-proof-based ceremony protocol that asymptotically outperforms the state-of-the-art~\cite{boneh-pot-2024}, as shown in~\cref{tb:fraud-cost}.

\begin{table}[b!]
\centering
\caption{Gas Cost with/out Fraud Proof Mechanism}
\label{tb:fraud-cost}
\resizebox{\columnwidth}{!}{ 
    \begin{threeparttable}
    \begin{tabular}{lcc}
    \toprule
    & Update Cost & Fraud Proof Verification Cost \\
    \midrule
    \cite{boneh-pot-2024} & ${\approx} 2n$ \var{ECMULT} + $4k - 2$ \var{ECPAIR} & $-$ \\
    \midrule
    \multirow{2}{*}{\sysname} & ${\approx}(n+k)$ \var{CALLDATA} & $\approx \log_{2} (n + k)$ \var{CALLDATA} \\
    & $+ 2(n+k)$ \var{Keccak} & $+ \log_{2} (n + k)$ $\var{Keccak}$ \\
    \bottomrule
    \end{tabular}
    \begin{tablenotes}\footnotesize
    \item[*] \var{ECMULT} and \var{ECPAIR} are ${>}10\times$ more expensive than \var{CALLDATA} and \var{Keccak}.
    \item[*] Aggregation is not considered in this comparision.
    \end{tablenotes}
    \end{threeparttable}
} 
\end{table}



\subsection{Aggregateable Contribution}
\label{sub:agg-contribution}

Despite the improvement above, the on-chain asymptotic costs remain 
$O(m\cdot d)$ (in terms of calldata) for $m$ contributors, as each contributor still needs to 
update the $\pp$ with proofs individually.
We aim to reduce the on-chain cost to $O(d)$, independent of $m$.
We introduce an aggregation technique that compresses multiple updates into a single one, forming our second key contribution to this work. 

Our second key observation to reduce cost of the proof-of-tau ceremony (c.f. \circled{2} in~\cref{fig:protocol-overview}) is as follows.
Instead of having each contributor upload their own update on-chain, 
we propose a (trustless) operator that aggregates contributors' updates and proofs \textit{off-chain}. This operator only needs to upload a pair consisting of the final $\pp$ and an aggregated proof. 
Upon submission, the contract performs $\checkone$ and $\checkthree$ but omits $\checktwo$ (as described in~\cref{subsec:fraud-proof-overview}). 

In our aggregation protocol, contributors first obtain the public parameters $\pp$ from the operator. After computing their proofs, they submit updates to the operator to incorporate into $\pp$. Once the operator publishes the updated $\pp$ on-chain, contributors can verify that their contributions were properly included. Only after this verification do contributors proceed with additional operator interactions.

\pparagraph{Aggregation Process} 
In more detail, the aggregation process occurs sequentially as follows: the first contributor queries the latest state $\pp_{t}$ from the operator. The contributor then updates $\pp_t$ to $\pp_{t+1}$ and submits the proof $\pi_{t+1}$ for $\checkone$ to the operator. After the operator verifies the update and the proof, the process continues with the next contributor.
Once the operator has received $m$ contributions, it will execute a proof aggregation algorithm to generate a single proof from all individual proofs received from the contributors. Both the final aggregated contribution and the aggregated proof are then submitted to the contract. 

\pparagraph{Key Challenge in Proof Aggregation} The main challenge in this process is how to design such an aggregated proof. Each contributor must prove their knowledge of a random value $r$ used to update $\pp = (\tau \cdot \genone, \ldots)$ to $\pp' = (r \cdot \tau \cdot \genone, \ldots)$, ensuring that they cannot ``reset'' $\tau$.
The contributor can use Schnorr's $\Sigma$-protocol~\cite{crypto-1989-1727} and BLS-style proof of possession~\cite{eurocrypt/RistenpartY07} to prove its knowledge of $r$.

However, when aggregating updates, the final $\pp'$ becomes $(r_{1} \ldots r_{m} \cdot \tau \genone, \ldots)$, and no party -- including the operator or contributors -- knows the accumulated random value $r_{1} \cdot \ldots \cdot r_{m}$.
Furthermore, for security reasons, no individual should learn the entire accumulated randomness.

\pparagraph{Alternatives and Limitations} One alternative is to use off-the-self MPC protocols or collaborative zk-SNARKS. These techniques, however, require the contributors to collaborate and stay online, executing the protocol synchronously. This introduces inefficiencies, as re-execution of the entire protocol may be needed if a party aborts.
Instead, we aim for a \emph{fire-and-forget} aggregation, where contributors can complete their tasks in a single message to the operator, avoiding complex interactions with either the operator or other contributors.

\pparagraph{Efficiency and Practicality}
Looking ahead, the succinctness of our aggregated proof and the efficiency of verification are crucial for practicality.
Our conservative estimation on the on-chain cost of the prior Schnorr's $\Sigma$-based-protocol~\cite{boneh-pot-2024} is ${>}10,000$ gas per contributor.
In Ethereum's KZG Ceremony for PoT, with over  $140,000$ contributors, 
this would lead to a gas consumption exceeding $1.4$ billion, more than $46\times$ the gas limit per Ethereum block.

In contrast, our scheme's verification is significantly more efficient.
The smart contract only needs to verify 
a proof consisting of a three group elements  using three pairing operations, making this the first practical approach for Ethereum's KZG ceremony for PoT. 

Next, we provide an overview of our aggregation protocol. For ease of exposition, we introduce our scheme gradually through different cases with increasing complexity.

\subsubsection{Strawman Aggregation Protocol}


In this part, we explain different strawman protocols. 

\pparagraph{Single Contributor Case}
Suppose there is only one contributor $\user_{1}$. It samples a secret key $\sk$ and public key $\pk$:
\begin{align*}
\sk \sample \Zp \quad\quad \pk = \sk \cdot \genone \in \gpone
\end{align*}
Then it computes a $\pp$ exactly like the normal case: it samples $r \sample \Zp$ and update $\pp$ to $\pp'$ using $r_{1}$ as the random value.
To produce a proof of possession of $r$, it computes
\begin{align*}
\sigma = \sk \cdot \ppArray{\pp}{Q_1} = \sk \cdot r \cdot \gentwo \in \gptwo
\end{align*}
and outputs $\pi = (\pk, \sigma)$ as the proof.

For a verifier, \eg, $\contract$ or $\server$, to conduct $\checkone$,
the contributor submits $\pp$, the verification key $\vk = \pk$, and $\sigma$ to the verifier, who then verifies that
\begin{align*}
e(\vk, \ppArray{\pp}{Q_1}) \isequal e(\genone, \sigma)
\end{align*}
Correctness holds because
\begin{align*}
 e(\vk, \pp.Q_{1}) = e(\sk \genone, r \gentwo) &=  e(\genone, (\sk \cdot r) \gentwo) 
 \\&= e(\genone, \sigma)
\end{align*}

Meanwhile, the secrecy of $r$ and possession of $r$ (by $\user_{1}$) are both gauranteed
by the $(n, k)$-Strong Diffie-Hellman Assumption.


\pparagraph{Two Contributors Case}
While the above construction already provides a new protocol for checking $\checkone$,
the most interesting aspect is how we can extend it to an aggregate protocol for multiple contributions.


Now, suppose there is a second contributor $\user_{2}$ with random value $r'$ and a key pair $(\sk', \pk') \in \Zp^* \times \gpone$.
Firstly, $\user_{2}$ obtains the prior $\pp$ from $\user_{1}$ and computes $\pp'$ from $\pp$ with $r'$.

This time, the computation for $\checkone$'s proof is different.
$\user_{2}$ also obtains $\sigma$ from $\user_{1}$,
computes
\begin{align}\label{eqt:single-update}
  \sigma' = r' \cdot \sigma + \sk' \cdot \ppArray{\pp'}{Q_1}
\end{align}
and finally outputs $\pp'$ with $\pi' = (\pk', \sigma')$ as the proof.

Anyone with $\pk$ can verify this proof by computing
a new verification key $\vk \leftarrow \pk + \pk'$ and checking
\begin{align}\label{eqt:single-pairing-check}
e(\vk, \ppArray{\pp'}{Q_1}) \isequal e(\genone, \sigma')
\end{align}
Correctness holds because
\begin{align*}
  e(\vk, \ppArray{\pp'}{Q_1}) &= e(\pk + \pk', r' r \gentwo)
  \\ &= e(\genone, (\sk + \sk') r' r \gentwo)
  = e(\genone, \sigma')
\end{align*}

Again, by the $(n,k)$-Strong Diffie-Hellman assumption, the above protocol guarantees
i) the secrecy of $r_{1}$ and $r_{2}$ , and ii) $\user$ and $\user'$ joint possession of $r' \cdot r$.

Meanwhile,
this protocol achieves most of the promised goals:
\emph{(i)} the proof is aggregated into a\emph{ co}nstant-size proof ($\sigma'$ and $\vk$), consisting of only two group elements;
\emph{(ii)} it has the \textit{fire-and-f\emph{orge}t} property, as $\user_{2}$ does not need to interact with $\user_{1}$ except for obtaining $\sigma$,
{(Looking ahead, a trustless operator will take care of the relay of $\sigma$s to entirely eliminate interaction among users)}.
and \emph{(iii)} it produces a proof for the \emph{oblivious} accumulated randomness, as neither $\user_{1}$ nor $\user_{2}$ knows $r' \cdot r$


\pparagraph{Three or More Contributors Case}
Generalizing the above strawman scheme to more contributors is straightforward.
Each new contributor $\user'$ retrieves $\pp$ and $\sigma$ from the previous contributor, updates $\pp$ to $\pp'$, and computes the new proof $\pi' = (\pk', \sigma')$ using \Cref{eqt:single-update}.
Verification of this update can be performed by anyone with access to all the public keys ($\pk$s), using \Cref{eqt:single-pairing-check} where $\vk$ is the sum of all $\pk$s.

\paragraph*{Adding Operators and Contract into the Picture.}
The role of an operator $\server$ is to coordinate the contribution in a batch.
It collects the public keys ($\pk$s) from the contributors and relays the latest $\pp$ and $\sigma$.
Once enough contributions have been gathered, the operator uploads the final $\pp$ along with the aggregated proof $\pi_{\server}$ to the contract $\contract$:
\begin{align*}
\footnotesize
\pi_{\server} &= \mathsf{agg}_{\mathsf{straw}}(\{\pi_{i} = (\sigma_{i}, \pk_{i})\}_{i \in [1..m]}) 
\\& \defn (\sigma_{m}, \vk = \sum_{i \in [i..m]} \pk_{i})
\end{align*}

The contract $\contract$ will verify $\pp$ with $\pi_{\server}$ using \Cref{eqt:single-pairing-check}, while maintaining the fraud-proof mechanism, and then update $\pp$.
After this, another operator can take over by fetching the $\pp$ and $\sigma$ from $\contract$.

\paragraph*{Maintaining $\vk$ on $\contract$.}
Notice that an operator might not have access to the $\pk$s held by previous operators, preventing it from computing the full verification key $\vk$ by summing all $\pk$s.
Our solution is to let $\contract$ store $\prevk$, the verification key from the previous update.
When a new operator updates $\pp$, it provides the set of $\{\pk_{i}\}_{i}$ by submitting $\curvk = \sum_{i} \pk_{i}$ to $\contract$. The contract then computes $\vk = \prevk + \curvk$.
Once the submission is verified, $\contract$ updates $\prevk \leftarrow \vk$.




\subsubsection{Contribution Inclusion Check}
A critical challenge arises in ensuring that contributors can verify whether their contributions were included in the final $\pp$. 
Even after observing $\pp$, $\sigma$, and $\vk$ on $\contract$, a contributor $\user_{i}$ has no guarantee that their contribution is indeed included.
The operator $\server$ could claim to accept $\user_{i}$'s contribution, but then discard it and continue with an independent $\pp$.
As a result, $\user_{i}$ would be unable to confirm their participation in the final $\pp$.

\pparagraph{Inclusion Check Scheme}
To resolve this, we need a protocol that assures $\user_{i}$ that their random value $r_{i}$ is indeed part of the final $\pp$, specifically that $\ppArray{\pp}{Q_{1}} = (r_{i} \cdot \ldots ) \gentwo$.
A simple solution is for $\server$ to publish a list of the public keys $\stkey = (\pk_{1}, \ldots, \pk_{m})$ allowing each $\user_{i}$ to verify that 
$\vk = \pk_{1} + ... + \pk_{m}$ and that $\pk_{i} \in \st_{\pk}$.

The intuition is that since $\vk$ contains $\user_{i}$'s $\sk_{i}$,
the final $\sigma$ should contain $\user$'s $\sigma = (\sk_{i} + \ldots) (r_{i} \cdot \ldots) \gentwo$ to pass the pairing equality check.
Since $\user$'s $\sigma$ contains $r_{i}$, the pairing equality check ensures that $\ppArray{\pp}{Q_{1}}$ also contains $r_{i}$.
More formally, the following chain of reasoning should hold:
{\small
\begin{align*}
  &e(\vk, \ppArray{\pp}{Q_{1}}) = e(\genone, \sigma) \text{ and } \vk = \sk_{i} \genone + \ldots \\
  \Rightarrow& e((\textcolor{red}{\sk_{i}} + \ldots) \genone, \ppArray{\pp}{Q_{1}}) = e(\genone, (\textcolor{red}{ \sk_{i}} + \ldots) (\textcolor{red}{r_{i}} \cdot \ldots) \gentwo) \\
  \Rightarrow & e((\sk_{i} + \ldots) \genone, (\textcolor{red}{r_{i}} \cdot \ldots) \gentwo) = e(\genone, (\sk_{i} + \ldots) (r_{i} \cdot \ldots) \gentwo) \\
  \Rightarrow & \ppArray{\pp}{Q_{1}} = (r_{i} \cdot \ldots)  \gentwo
\end{align*}
}

\pparagraph{(Intra-Operator) Rogue Key Attack}
Despite the inclusion check protocol, there remains a vulnerability known as the \emph{Rogue Key Attack}, a well-known issue in BLS-style schemes.
In our context, a rogue key attack allows a malicious operator to deceive a contributor into believing that its randomness is included in the latest $\pp$.

Here is how the attack works in our context:
Suppose a contributor, $\user$, with public key $\pk$, submits their key to the operator $\server$, who then maliciously discards all other information from $\user$.
$\server$ then proceeds to work with a compromised contributor $\user'$, with $(\sk'$, $\pk')$, and uploads the final $\pp$ containing only $\user'$'s contribution.

The goal is to make $\user$ believe that its contribution was included.
$\server$ publishes a list of public keys, $\stkey = (\pk, \pk' - \pk)$ and claims the verification key $\vk = \pk + (\pk' - \pk)  = \pk'$.
Notice how $\user$'s $\pk$ is effectively ``cancelled'' out.
Using $\sk'$, $\server$ generates a valid signature $\sigma = \sk' \cdot \ppArray{\pp}{Q_{1}}$, which passes the verification, thus falsely convincing $\user$ that their contribution is present simply because $\pk$ appears in $\stkey$.



\pparagraph{Preventing Rogue Key attack with Proof-of-Possession (PoP)}
To defend against this attack, we adopt the proof-of-possession approach~\cite{eurocrypt/RistenpartY07}. This requires each contributor to prove ownership of their public key by signing a hash of the key using their secret key.
Specifically, each contributor $\user$ must provide $\pos = \sk \cdot \hashgp(\pk)$ to $\server$, where the operator can verify the proof using the pairing equation $e(\pk, \hashgp(\pk)) =_{?} e(\genone, \pos)$.
The key insight is that an adversary generating fake public keys (like $\pk' - \pk$) cannot sign on the corresponding hash without knowing the respective secret keys ($\sk' - \sk$).

\pparagraph{Our Final Inclusion Check Protocol}
We enhance the inclusion check protocol by incorporating a PoP mechanism.
$\server$ now publishes the list of key and proof pairs, $\stkey = \{(\pk_{i}, \pos_{i})\}_{i}$ on its public server.
During the verification process, a contributor $\user_{j}$ with public key $\pk_{j}$ confirms:
\emph{(i)} $(\pk_{j}, *) \in \stkey$;
\emph{(ii)} $\curvk = \sum_{i} \pk_{i}$;
and \emph{(iii)} the validity of all $\pos$s.

With these checks in place, the protocol secures the inclusion of each contributor's randomness, and we formally prove its security in \cref{thm:incl-snd}.

\subsubsection{Progressive Key Aggregation}\label{sec:pro-key-agg}
The proof of possession approach can prevent \emph{intra-operator} rogue key attacks,
where the rogue keys exist among an operator's collected $\pk$s.
However, it does not eliminate the \emph{inter-operator} rogue key attacks, where the rogue keys come from other operators.



\pparagraph{(Inter-Operator) Rogue Key Attacks}
Here, we illustrate this attack.
Suppose the (honest) first operator $\server_{1}$
has uploaded $\pp$ and $\vk_{1}$ with its proof.
Now, $\prevk = \vk_{1}$.
A malicious second operator $\server_{2}$ with $(\sk_{2}, \pk_{2})$ can then perform a ``key cancellation'' to remove $\vk_{1}$ from the next verification key. 
It simply uploads $\curvk = -\vk_{1} + \pk_{2}$.
Notice, the new verification key to check $\server_{2}$'s update becomes
\begin{align*}
    \vk = \prevk + \curvk = \vk_{1} + (-\vk_{1} + \pk_{2}) = \pk_{2}
\end{align*}
Now, $\server_{2}$ can use its $\sk_{2}$ to produce a valid $\sigma$ on any $\pp$.

\pparagraph{Unhelpful PoP}
Unfortunately, the PoP approach, which defends against the intra-operator rogue key attack, is ineffective here.
Verifying all $\pk$s and $\pos$s would be too costly for $\contract$, as their total size scales linearly with the number of contributions, thus defeating the purpose of a constant-size aggregate proof.

\pparagraph{Random Coefficients for $\vk$s and $\sigma$s}
To defend this attack, we resort to a ``random coefficient'' approach~\cite{ctrsa/FerraraGHP09,eurocrypt/CamenischHP07,asiacrypt/BonehDN18},
which is a recurring technique for batched signature verification.
Instead of computing $\vk = \prevk + \curvk$,
$\contract$ now samples two random coefficients $\rho_{1}, \rho_{2} \sample \Zp^{*}$
and computes $\vk = \rho_{1} \prevk + \rho_{2} \curvk$.
Now, the key cancellation trick cannot help a malicious $\server$ to get rid of the prior $\prevk$
because with overwhelming probability
\begin{align*}
  \vk &= \rho_{1} \prevk + \rho_{2} \curvk = \rho_{1} \vk_{1} + \rho_{2} (-\vk_{1} + \pk_{2})
  \\&= (\rho_{1} - \rho_{2}) \vk_{1} + \ldots \text{ and } (\rho_{1} - \rho_{2}) \neq 0
\end{align*}

Obviously, $\sigma$ computed by our prior approach cannot pass \Cref{eqt:single-pairing-check}
as $\vk$ is no longer a simple sum of all $\pk$s.
To cater for the randomly combined $\vk$s, the operator and contributors need to instead update two $\sigma$s: $\presig$ and $\cursig$.
In our protocol, $\server$ fetches $\presig$ from $\contract$ and then asks $\user$ with a random value $r$, new $\pp$, and key pair $(\sk, \pk)$ to update as following
\begin{align*}
  \presig \leftarrow r \cdot \presig \quad \cursig \leftarrow \sk \cdot \ppArray{\pp}{Q_{1}} + r \cdot \cursig
\end{align*}
Note that $\presig$ and $\cursig$ are supposed to contain the $\sk$s in $\prevk$ and $\curvk$, respectively.

Then, $\user$ sends $\pp, \pi = (\presig, \cursig, \pk, \pos)$ to $\server$.
After gathering all the updates from contributors, 
$\server$  uploads to $\contract$ the final $\pp$ and 

{
\small
\begin{align*}
\pi_{\server} = \mathsf{agg}(\{\pi_i\}_{i \in [m]}) \defn (\ppArray{\pi_m}{\presig}, \ppArray{\pi_m}{\cursig}, \curvk = \sum_{i \in [m]} \pk_i)
\end{align*}
}


Finally,
$\contract$ computes $\vk = \rho_{1} \prevk + \rho_{2} \curvk$ and $\sigma = \rho_{1} \presig + \rho_{2} \cursig$ and
uses \Cref{eqt:single-pairing-check} to verify them.
The correctness and soundness hold because the above verification is (almost) equivalent to verifying
{\small 
\begin{align*}
e(\prevk, \ppArray{\pp}{Q_{1}}) \isequal e(\genone, \presig)
  \wedge
e(\curvk, \ppArray{\pp}{Q_{1}}) \isequal e(\genone, \cursig)
\end{align*}
}

and for each of these equality checks, we can apply our prior conclusion about the security properties.

\pparagraph{Progressive Key Aggregation}
The random combination on $\sigma$ and $\vk$ is \emph{not} only for this one-time verification.
After passing the verification, $\contract$ will update in its state:
\begin{align*}
  \prevk \leftarrow \rho_{1} \prevk + \rho_{2} \curvk ; \quad
  \presig \leftarrow \rho_{1} \presig + \rho_{2} \cursig
\end{align*}
Without these random coefficients in the update $\prevk$, an inter-operator rogue key attack is still possible.

Another issue for $\contract$ is how to sample $\rho_{1, 2}$.
In practice, $\contract$ is implemented on a blockchain and $\contract$'s randomness can be biased to favor the adversary.
Here, we propose $\contract$ samples $\rho_{\{1,2\}}$ via the following hashing
\begin{align*}
\rho_{j} = H(\prevk ||\curvk || \presig || \cursig || \ppArray{\pp}{Q_{1}} || j ) \text{ for } j \in \{1, 2\}
\end{align*}

We provide a formal security analysis and proofs
against inter-operator rogue key attacks in Theorem~\ref{thm:incl-snd} (Lemma~\ref{lem:hash-zippel} and Lemma~\ref{lem:snd-after-t}).



\section[Construction]{Construction and Security Analysis}
\label{sec:approach}



In this section, we formalize our ideas into algorithms and analyze the security of our proposed construction.

\begin{algorithm}[t]
\footnotesize
    \caption{The On-Chain Update}
    \label{alg:on-update}
    \begin{algorithmic}[1]
    
    
    
    \Procedure{$\fun{ContractSetup}$}{$n, k$}
    \State
    Sets $\pp \leftarrow (\genone, \ldots, \genone; \quad \gentwo, \ldots, \gentwo ) \in \gpone^{n} \times \gptwo^{k}$
    \State
    Sets $\prevk \leftarrow \genone, \presig \leftarrow \gentwo, \arrcompp = ()$
    \State
    Sets $\ell \leftarrow 0$ \Comment{The number of current rounds}
    \EndProcedure
  
    \Procedure{$\fun{OnUpdate}$}{}
    \Comment{Invoked by $\server$}
    \State Outputs $\fun{ContractUpdate}(\pp, \presig', \cursig', \curvk) \neq \mathsf{Abort}$
    \EndProcedure

    \Procedure{$\fun{ContractUpdate}$}{$\pp, \presig', \cursig', \curvk$}
  
    \LComment{Executed on $\contract$}
  \State Computes
          $\rho_{j} = \hashzp(\prevk||\curvk||\presig'|| \cursig' || \pp.Q_{1} || j)$ for
          $j \in \{1, 2\}$
  \State Computes $\vk^{*} = \rho_{1} \cdot \prevk + \rho_{2} \cdot \curvk $
  \State Computes $\sigma^{*} = \rho_{1} \cdot \presig + \rho_{2} \cdot \cursig $
  \If{$e(\genone, \sigma^{*}) \neq e(\vk^{*}, Q_{1})$ or $\pp.P_{1} = \genone$}
  \State Aborts
  \EndIf
  \State Sets $\prevk \leftarrow \vk^*, \presig = \sigma^*$
  \State Sets $\arrcompp \leftarrow \arrcompp || \com(\pp), \ell \leftarrow \ell + 1$
    
    \EndProcedure
  
  \end{algorithmic}
\end{algorithm}


{\centering
\begin{algorithm}[t]
\footnotesize
\caption{\small The Off-Chain Update}
\label{alg:off-update}
\begin{algorithmic}[1]


\Procedure{$\fun{ServerSetup}$}{}
\State
Fetches $\prevk$ from $\contract$'s state.
\State
Sets $\pp = (\genone, \ldots, \genone; \quad \gentwo, \ldots, \gentwo )$
\State
Sets
$\presig \leftarrow \gentwo, \cursig \leftarrow \gentwo, \curvk \leftarrow \genone, \stkey \leftarrow \{\}$
\EndProcedure

\Procedure{$\fun{OffContribute}$}{}
  \Comment{Invoked by $\user$}
\State
Fetches
$(\pp, \presig, \cursig)$ to from $\server$'s state $\stser$
\State Parses $\pp$ as $(P_{1}, \ldots, P_{n}; \quad Q_1, \ldots, Q_k )$
\If{$P_1 = \genone$}
Outputs $0$
\EndIf
\label{algl:cursig-zero}
\State Samples $r \sample \Zp^{*}$ and $\sk \sample \Zp^{*}$
\State Computes $\pp' = (r P_{1}, \ldots, r^{n} P_{n}; \quad r Q_1, \ldots, r^{k} Q_k )$
\State Computes $\pk = \sk \cdot \genone$ and $\pos = \sk \cdot \hashgp(\pk)$.
\State Computes $\presig' = r \cdot \presig$ and
$\cursig' = r \cdot \cursig + \sk \cdot r \cdot P_{1}$
\State Sets $\pi \leftarrow (\pk, \presig', \cursig', \pos)$
\State Outputs $\fun{ServerOffUpdate} (\pp', \pi)$
\EndProcedure






\Procedure{$\fun{ServerOffUpdate}$}{$\pp', \pi$}
  \Comment{Executed on $\server$}
  \State Computes $\vk' = \curvk + \pi.\pk$
  \State Computes $(b, \cdot, \cdot) = \checkwellform(\pp)$
  \If{$b = 0$
  or $\pp'.P_{1} = \genone$ or $e(\pi.\pos, \gentwo) \neq e(\pk, \hashgp(\pi.\pk))$
  or $e(\pi.\presig', \gentwo) \neq e(\prevk, \pp'.Q'_{1})$
  or $e(\pi.\cursig', \gentwo) \neq e(\vk', \pp'.Q'_{1})$
  or $(\pk, \cdot) \in \stkey$}
\State Outputs $0$
\EndIf
\State Updates $\pp \leftarrow \pp'$, $\presig \leftarrow \pi.\presig'$, $\cursig \leftarrow \pi.\cursig'$
\State Updates $\curvk \leftarrow \pi.\vk', \stkey \leftarrow \stkey || (\pi.\pk, \pi.\pos)$
\State Outputs $1$
\EndProcedure

\end{algorithmic}
\end{algorithm}
}

\subsection{Construction}
Our construction consists of the following procedures.




\begin{compactitem}
    \item $\fun{ContractSetup}(1^\lambda, n, k)$ takes as input the degrees of the PoT string $n, k$ and implicitly the security parameter $1^{\lambda}$. Once this algorithm is executed, $\contract$ is considered deployed on a permissionless blockchain and able to receive PoT update for the ceremony.
    The procedure is specified in \Cref{alg:on-update}.
    
    \item $\fun{ServerSetup}(1^{\lambda})$ implicitly takes as input the security parameter $1^{\lambda}$.
        This function sets an operator by fetching data from $\contract$ and initializing a public state for off-chain update.
    The procedure is specified in \Cref{alg:off-update}.

    \item $\fun{OffContribute}$ is invoked by and implicitly receives randomness from a contributor.
        This function enables a contributor to perform an off-chain update on $\pp$ by communicating with $\server$ and invoking $\server$'s $\fun{ServerOffUpdate}$ to provide the updated $\pp$ and proof. Finally, the contributor will receive a bit $0/1$ indicating if the off-chain update is successful.
    The procedures are specified in \Cref{alg:off-update}.

    \item $\fun{OnUpdate}$ is invoked by an operator after gathering off-chain updates.
        This function submits the updated $\pp$ with the proof to $\contract$ via invoking $\contract$'s $\fun{ContractUpdate}$.
        If $\contract$ does not abort, this function considers the update successful and outputs $1$; otherwise, $0$.
    The procedure is specified in \Cref{alg:on-update}.

    \item $\fun{CheckInclusion}(\pk)$ is invoked by $\user$ with a $\pk$ after $\server$ has completed an on-chain update.
        This function fetches data from $\server$'s public state and from the state and transaction history of $\contract$,
        and outputs a bit $0/1$ indicating if $\user$'s contribution is included on a valid on-chain update.
    The procedure is specified in \Cref{alg:check-inclusion}.

    \item $\fun{Disprove}(t)$ is invoked by $\user$ and takes as input a round number $t$.
        This function communicates with $\contract$ and submits a fraud-proof via invoking $\contract$'s $\fun{RecvFraud}$.
        It outputs a bit $0/1$ indicating if $\contract$ has rewound the state to void that $\pp$.
    The procedure is specified in \Cref{alg:give-fraud-proof}.

\end{compactitem}

{\centering 
\begin{algorithm}[t]
    \footnotesize
    \caption{The Inclusion Check}\label{alg:check-inclusion}
    \begin{algorithmic}[1]
    
    
    
    \Procedure{$\fun{CheckInclusion}$}{$\pk$}
    \Comment{Invoked by $\user$}
    
    \State
    Fetches
    $\pp^{\server}, \curvk^{\server}, \stkey$ from $\server$
    \State
    Fetches
    $\pp^{\contract}, \curvk^{\contract}$
    from $\contract$'s transaction
    
    \State
    Computes $(b, \cdot, \cdot) = \checkwellform(\pp)$.
    \If{$b = 0$}
    Outputs $0$
    \EndIf
    
    \State
    Sets $\vk' = \genone$
    \For{each $(\pk_{i}, \pos_{i})$ in $\stkey$}
        \If{$e(\pk_{i}, \hashgp(\pk_{i})) \neq e(\genone, \pos_{i})$}
        Outputs $0$
        \EndIf
      \State Sets $\vk' = \vk' + \pk_{i}$
    \EndFor
      %
    \If{$\vk' = \curvk^{\mathsf{ser}} = \curvk^{\mathsf{tx}}$ and $\pp^{\server} = \pp^{\contract}$ and $(\pk, \cdot) \in \stkey$}
 Outputs $1$
    \Else \xspace Outputs $0$
    \EndIf
    \EndProcedure
    
\end{algorithmic}
\end{algorithm}
}

\begin{algorithm}[t]
    \footnotesize
    \caption{The Fraud Proof Algorithm}\label{alg:give-fraud-proof}
    \begin{algorithmic}[1]
    %

    \Procedure{$\fun{Disprove}$}{$t$}
    \Comment{Invoked by $\user$}
    \State
    Read $(t, \pp_{t}, \compp_{t})$ from $\contract$'s $\txcontract^{t}$
    \State
    Read $\compp^{\st}_{t}$ from $\contract$'s $\stcontract$
    \If{$\compp_{t} \neq \compp^{\st}_{t}$}
       Aborts
    \EndIf




    \State Computes $(\idxgp, \idxill) = \funextwellform(\pp_t)$.
    \State $\mathcal{M} = \memprove(\pp_{t}, ((1, 1), (2, 1), (\idxgp, \idxill{-}1), (\idxgp, \idxill)))$
    \LComment{$\memprove$ ignores $\idxgp$ and $\idxill$ if they are $\bot$}



    \State
    Outputs $(t, \idxgp, \idxill, \mathcal{M})$

    \EndProcedure
    

    \Procedure{$\fun{RecvFraud}$}{$t, \idxgp, \idxill, \mathcal{M}$}
    \Comment{Invoked by $\contract$}
    \If{$\fun{FraudVerify}(t, \idxgp, \idxill, \mathcal{M}) = 0$}
    Aborts
    \Else
    \xspace
    $\arrcompp \leftarrow \arrcompp[1:t-1], \ell \leftarrow \ell-1$
    \LComment{Removes the items after the $t$-th ones.}
    \EndIf
    \EndProcedure
    
    \Procedure{$\fun{FraudVerify}$}{$t, \idxgp, \idxill, \mathcal{M}$}
    \Comment{Invoked by $\contract$}
    
    \State
    Extracts
    $\arrelem = (\pp[1][1], \pp[2][1], \pp[\idxgp][\idxill{-}1], \pp[\idxgp][\idxill])$ from $\mathcal{M}$
    \LComment{Ignores $\pp[\idxgp][\idxill{-}2], \pp[\idxgp][\idxill]$ if $\idxgp = \bot$}
    \If{any item in $\arrelem$ cannot verify with $\arrcompp[t]$ and $\mathcal{M}$}
    Outputs $0$.
    \EndIf
    \State
    Names $P_{1} = \pp[1][1], Q_{1} = \pp[2][1]$
    \If{$e(P_{1}, \gentwo) \neq e(\genone, Q_{1})$}
     Outputs $1$.
    \EndIf
    
    \State
    Names $E_{i-1} = \pp[\idxgp][\idxill{-}1], E_{i} = \pp[\idxgp][\idxill]$
    \If{$\idxgp = 1$ }
     Outputs $e(E_{i}, \gentwo) = e(E_{i-1}, Q_{1})$
    \ElsIf{$\idxgp = 2$}
     Outputs $e(\genone, E_{i}) = e(P_{1}, E_{i-1})$
    \Else
     \xspace Outputs $0$
    \EndIf
    \EndProcedure

\end{algorithmic}
\end{algorithm}


{\centering
\begin{algorithm}[!ht]
    \footnotesize
    \caption{Extended Well-Formedness Check}\label{alg:local-ext-wellform}
    \begin{algorithmic}[1]
      \Procedure{$\checkwellform$}{$\pp$}
      \State Parses $\pp$ as $(P_{1}, \ldots, P_{n}; \quad Q_{1}, \ldots, Q_{k})$

      \If{$e(P_{1}, \gentwo) = e(\genone, Q_{1})$}
         Outputs $(1, 1, 1)$
      \EndIf

      \For{each $i \in [2..n]$}
        \If{$e(P_{i}, \gentwo) \neq e(P_{i-1}, Q_{1})$}
          Outputs $(1, 1, i)$.
        \EndIf
      \EndFor

      \For{each $i \in [2..k]$}
        \If{$e(\genone, Q_{i}) \isequal e(P_{i-1}, Q_{1})$}
           Outputs $(1, 2, i)$
        \EndIf
      \EndFor

      \State Outputs $(0, \bot, \bot)$ and halts.
    \EndProcedure
    \end{algorithmic}
\end{algorithm}
}

Finally, \Cref{fig:dml-mixer} 
summarizes who initiates each algorithm, the execution order, and the interaction during the algorithms' execution.

\begin{figure}[!ht]
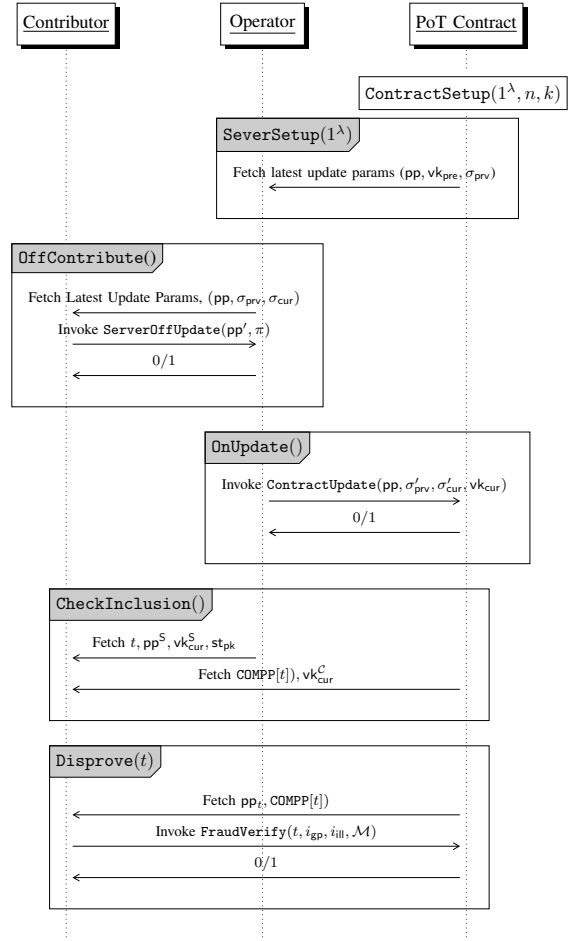

    \centering
    \resizebox{.95\columnwidth}{!}{
    \begin{sequencediagram}
    \newinst{a}{Contributor}
    \newinst[2]{m}{Operator}
    \newinst[2]{b}{PoT Contract}
    \IMess{m}{}{b}
    \node (B) [fill=white] at (mess to){
    \pcbox{\mathtt{ContractSetup}(1^{\lambda}, n,k)}
    };
    \begin{sdblock}[black!20]{$\mathtt{SeverSetup}(1^{\lambda})$}{}
    \Mess{b}{Fetch latest update params ($\pp, \prevk, \presig$)}{m}
    \end{sdblock}
    \begin{sdblock}[black!20]{$\mathtt{OffContribute}$()}{}
    \Mess{m}{Fetch Latest Update Params, $(\pp, \presig, \cursig)$}{a}
    \Mess{a}{Invoke $\mathtt{ServerOffUpdate}(\pp', \pi)$}{m}
    \Mess{m}{$0/1$}{a}
    \end{sdblock}
    \begin{sdblock}[black!20]{${\mathtt{OnUpdate}()}$}{}
    \Mess{m}{Invoke $\mathtt{ContractUpdate}(\pp, \presig', \cursig', \curvk)$}{b}
    \Mess{b}{$0/1$}{m}
    \end{sdblock}
    \begin{sdblock}[black!20]{$\mathtt{CheckInclusion}()$}{}
    \Mess{m}{Fetch $t, \pp^\server, \curvk^\server, \stkey$}{a}
    \Mess{b}{Fetch $\arrcompp[t]), \curvk^\contract$}{a}
    \end{sdblock}
    \begin{sdblock}[black!20]{$\mathtt{Disprove}(t)$}{}
    \Mess{b}{Fetch $\pp_{t}, \arrcompp[t])$}{a}
    \Mess{a}{Invoke $\fun{FraudVerify}(t, \idxgp, \idxill, \mathcal{M})$}{b}
    \Mess{b}{$0/1$}{a}
    \end{sdblock}
    \end{sequencediagram}
    }
    \caption{The interaction among the entities in our protocol and the supposed invocation sequence of the procedures.} 
    \label{fig:dml-mixer}
\end{figure}


\subsection{Security Analysis}\label{subsec:security-analysis}

\pparagraph{Completeness of the Fraud Proof}\label{def:wellform-complete}
The following theorem states that if some $\pp_{t}$ is ill-formed, anyone can generate a fraud-proof from $\pp_{t}$ to rewind $\contract$ back to stage $t-1$.

\begin{theorem}
Suppose $\pp_{t}$ is submitted to $\contract$ at stage $t$ and $\pp_t$ is not in the form of
\begin{align*}
(\tau \genone, \tau^2 \genone, \ldots, \tau^{n} \genone;
\quad \tau \gentwo, \ldots, \tau^{k} \gentwo)
\end{align*}
for $\tau \in \Zp^{*}$.
There exists a proof $\pi$ derived from $\pp_{t}$ such that $\contract$ will rewind itself back to stage $t-1$ after receiving $\pi$ via $\fun{FraudProve}$.
\end{theorem}

The proof follows directly from the analysis in~\Cref{subsec:fraud-proof-overview}.

\paragraph{Soundness for the Fraud Proof}
The following states that if $\pp_t$ is well-formed, $\contract$ will not be rewinded back to $t-1$ with overwhelming probability.
\begin{theorem}
\label{thm:wellform-sound}
Suppose $\pp_{t}$ is submitted to $\contract$ at stage $t$ and $\pp_t$ is not in the form of
\begin{align*}
(\tau \genone, \tau^2 \genone, \ldots, \tau^{n} \cdot \genone;
\quad \tau \gentwo, \ldots, \tau^{k} \gentwo)
\end{align*}
for $\tau \in \Zp^{*}$.
With negligible probability, there exists a proof $\pi$ derived from $\pp_{t}$ such that $\contract$ will rewind itself back to stage $t-1$ after receiving $\pi$ via $\funfraudproof$.
\end{theorem}

\pparagraph{Soundness of Randomness Inclusion}
\label{subsec:inclusion-sound}
We consider an adversary $\adv$ that attempts to convince a contributor $\user$ that $\user$'s $r$ is included in the final $\pp$,
while $\adv$ has excluded $r$ from that $\pp$.
We define inclusion soundness by stating the above situation is computationally infeasible:
A PPT $\adv$ cannot correctly guess the trapdoor in the final $\pp$ with non-negligible probability as long as $\user$ has confirmed its randomness inclusion via $\fun{CheckInclusion}$.
We allow $\adv$ to
compromise all other contributors except $\user$ and all operators
and
invoke any function except $\fun{ContractSetup}$.
Nevertheless, we assume $\adv$ cannot alter the execution on $\contract$, and all invalid $\pp$s would be disproved before a final $\pp$ is decided.

\begin{definition}[$\mathsf{Game}_{\mathsf{incl\text{-}snd}}$]
The security game $\mathsf{Game}_{\mathsf{incl\text{-}snd}}$ for soundness of randomness inclusion is as follows:
\begin{compactenum}
    \item The challenger runs $ \mathsf{ContractSetup}(1^{\lambda})$ to setup $\contract$.
    \item For polynomially many times, $\adv$ can invoke $\fun{ServerSetup}$ and $\fun{OnUpdate}$ as an operator, $\fun{OffUpdate}$ and $\fun{CheckInclude}$ with any operator, $\fun{ProveFraud}$ (as a contributor), and access the random oracle $\hashgp$ and $\hashzp$.
    \item The challenger invokes $\fun{OnUpdate}(\cdot)$ as a contributor with $\adv$ acting as an operator. Then, the challenger records its $\pk$.
  \item $\adv$ can invokes the above functions and oracles polynomially many times.
    Finally,
    $\adv$ outputs a state of an operator $\st_{\server}$, a transaction $\tx_{\contract}$, and a guess $w \in \Zp$.
    \item The challenger fetches the final $\pp$ from (the final valid transaction of $\fun{OnUpdate}$ to) $\contract$ and checks if
    \begin{compactitem}
        \item  $\tx_{\contract}$ is a valid transaction on $\contract$
        \item $\checkinclusion(\pk) \isequal 1$ where data are read from $\st_{\server}$ and $ \txser$
        \item $\pp \isequal (w \genone, \ldots w^{n} \genone; \quad w \gentwo, \ldots, w^{k} \gentwo)$
    \end{compactitem}
    If all checks are passed, it outputs $1$, or $0$ otherwise.
\end{compactenum}
\end{definition}


\begin{definition}[Soundness of Randomness Inclusion]
  \label{def:inclusionsound}
  A powers-of-tau system has inclusion soundness
  if
  for any $n, k \in \poly(\lambda)$, any PPT adversary with $\numhashquery = \poly(\lambda)$ queries to $\hashgp$ and $\hashzp$ can win $\mathsf{Game}_{\mathsf{incl\text{-}snd}}$ with only negligible probability.
\end{definition}

\begin{theorem}\label{thm:incl-snd}
Under $(n, k)$-SDH assumption,
our scheme is $\mathsf{Game}_{\mathsf{incl\text{-}snd}}$-secure in the AGM with random oracle model.
Specifically, we have
\begin{align*}
  \Pr[\gamesdh^{\advb} = 1] &\geq \frac{2p - \numhashquery}{12p} \Pr[\mathsf{Game}_{\mathsf{incl\text{-}snd}} = 1]
\end{align*}
where $p$ is the order of the pairing groups and $q_{H}$ is the number of query to the random oracles.
\end{theorem}


The proof of this theorem is also at the center of our contribution.
Due to the page limit, we defer all full proofs to~\Cref{sec:missing-proofs}.

\begin{proof}[Proof sketch]
If there is an algebraic adversary $\advinclude$ who makes a correct guess on the $\tau$ from the final $\pp$ with non-negligible probability into breaking, we construct another adversary $\advb$ whose use $\advinclude$ to break $(n,k)$-SDH assumption.
$\advb$ embeds the $(n, k)$-SDH instance into three places:
\begin{compactitem}
\item the $\pk$ of the contributor
\item the contributor's $r$ in the $\pp$ 
\item the responses of $\hashgp(\cdot)$
\end{compactitem}
And we show that these embeddings are sufficient to capture all the attacks in $\inclusionsoundgame$.

It means that to break the soundness of our aggregation scheme,
$\advinclude$ must at least correctly guess the $\sk$ from $\pk$, the $r$ from $\pp$, or breaking the PoP with non-negligible probability.
We also apply Schwartz-Zippel Lemma~\cite{ipl/DemilloL78, eurosam/Zippel79, jacm/Schwartz80} to show that $\adv$ can only break our progressive key aggregation with negligible probability.
\end{proof}

\section{Evaluation}
\label{sec:evaluation}

\begin{table}[t]
\centering
\caption{Gas costs for different operations}
\label{tb:evm-op-cost}
\small  
\begin{tabularx}{\columnwidth}{lll}
\toprule
Name         & Operation                                                             & Gas cost                   \\ \midrule
\var{ECADD}  & \(A + B \text{ for } A, B \in \mathbb{G}_1\)                          & 150                        \\
\var{ECMult} & \(\alpha A \text{ for } \alpha \in \mathbb{Z}_p, A \in \mathbb{G}_1\) & 6,000                      \\ 
\var{ECPAIR} & \(\sum_{i=1}^k e(A_i, B_i) = 0 \)                                     & \(34,000 \cdot k + 45,000\) \\ 
\var{CALLDATA} & \(k\) non-zero bytes                                                & \(16 \cdot k\)             \\ 
\var{Keccak256} & \(k\) 32-byte words                                                & \(30 + 6 \cdot k\)         \\ \bottomrule
\end{tabularx}
\end{table}



In this section, we implement \sysname and demonstrate its enhancements compared to the construction proposed by Nikolaenko et al.~\cite{boneh-pot-2024} for the on-chain protocol.

\pparagraph{Choices of Cryptographic Primitives} For our setup, we selected the BN254 curve, which is natively supported by Ethereum as specified in EIP-197~\cite{eip197}. For the hash function, we used \texttt{Keccak256}.

\pparagraph{Software} We implemented all off-chain components (i.e., Algorithms 2-5) of our protocol in \texttt{Rust}, and the on-chain component (i.e., Algorithm 1: the On-chain Update)  as a smart contract in \texttt{Solidity}, and we deployed the smart contract using \texttt{Ganache}~\cite{ganache} which is a sandbox environment for Ethereum.

\pparagraph{Gas Cost Accounting}
With all our optimizations in place, the gas cost for submitting a
contribution is mainly $O(n+k)$ \var{CALLDATA} (for reading $\pp$ in
transaction inputs) and \var{Keccak256} (for hashing).
The gas cost for verifying a fraud-proof is $O(\log\degree)$ \var{CALLDATA} and
$\var{Keccak256}$ for membership proofs, plus $2$ \var{ECPAIR}.
The gas cost for these operations is mentioned in~\cref{tb:evm-op-cost}.



\pparagraph{Hardware} We conducted our experiments on a desktop machine equipped
with a 2.6GHz 6-Core Intel Core i7 and 80 GB of RAM.

\subsection{Performance}



\pparagraph{On-chain Fraud-Proof Mechanism: On-chain costs} The on-chain gas costs are detailed in~\cref{tab:gas-cost}. Our protocol significantly reduces the update complexity by replacing the expensive check 2 with a fraud-proof mechanism. This optimization streamlines the update process to just two checks and the construction of a Merkle tree for the string. As a result, our update operation is 16 times more gas-efficient compared to the method proposed in~\cite{boneh-pot-2024}. This enhanced efficiency allows our protocol to support for the first time in Ethereum a string size of $2^{15}$, making it viable for a wide range of applications, including ProtoDanksharding.

Furthermore, in the event that a fraud-proof is necessary, the associated cost is fairly low at only 330,000 gas (approximately US\$2 at current rates). This cost-effectiveness in both regular updates and fraud-proofs underscores the practical applicability of our protocol in blockchain environments where gas optimization is crucial.

We compare the gas consumption of our scheme and Nikolaenko~\etal~\cite{boneh-pot-2024}'s in~\Cref{tab:gas-cost}.
And we plot the gas cost for different parameter sizes in~\Cref{fig:gas-plot}.

\begin{figure}[t]
    \centering
    \includegraphics[scale=0.19]{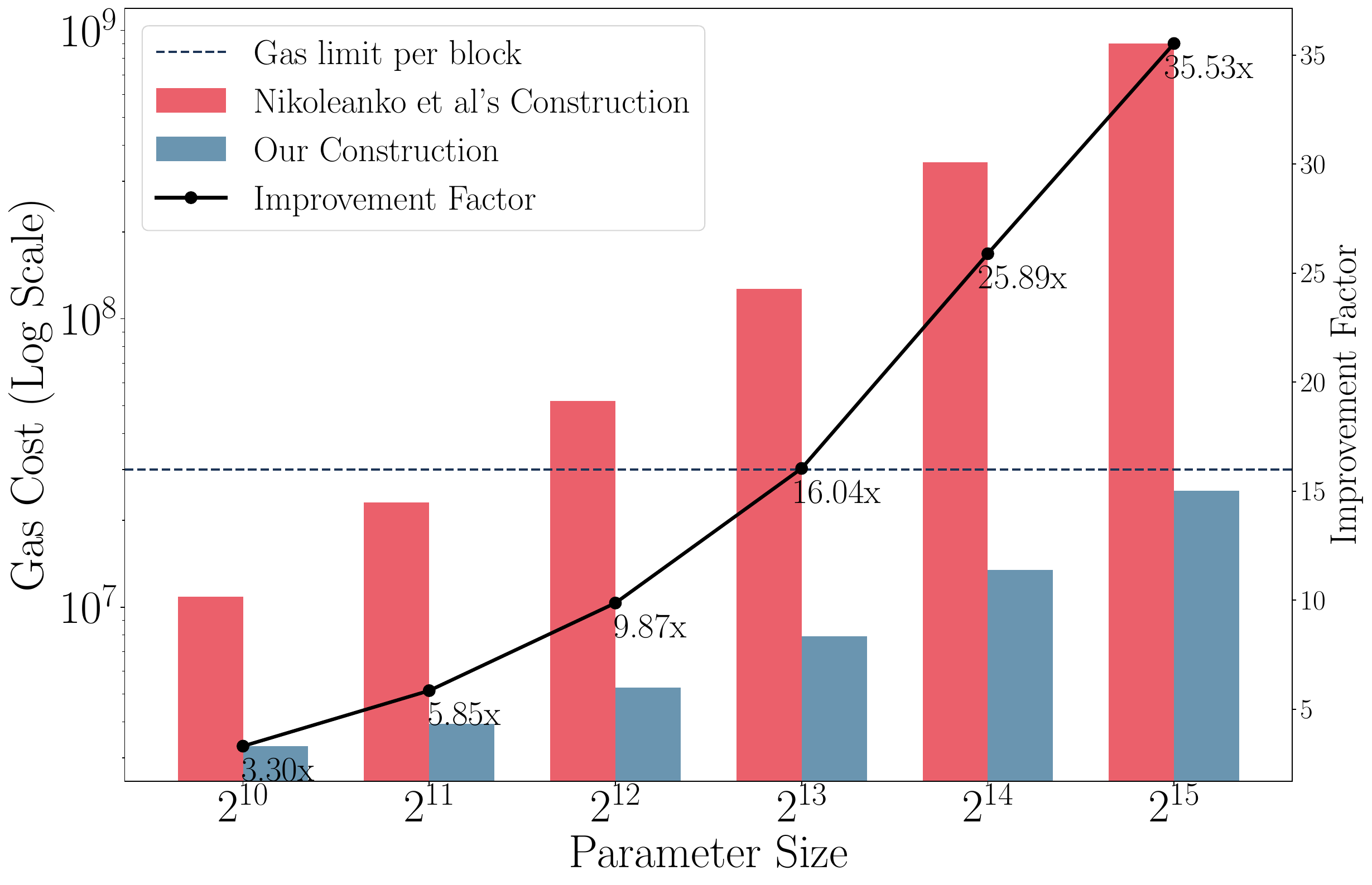}
    \caption{Comparison of Gas Costs for Each Update with~\cite{boneh-pot-2024} for $n\in \{2^{10}, 2^{11}, 2^{12}, 2^{13}, 2^{14}, 2^{15}\}$ and $k = 1$
    }
    \label{fig:gas-plot}
  \end{figure}

\begin{remark}
As previously noted, Ethereum does not natively support operations in group $\mathbb{G}_2$. Hence, the protocol in~\cite{boneh-pot-2024} is ineffective if the number of elements in $\mathbb{G}_2$ exceeds one (i.e., $k>1$). However, since our approach utilizes the blockchain for storage, we do not require operations within the group $\mathbb{G}_2$. Our fraud check mechanism relies only on pairings, allowing our protocol to support parameters with multiple elements in group $\mathbb{G}_2$.
\end{remark}

\begin{figure}[b]
    \centering
    \includegraphics[scale=0.175]{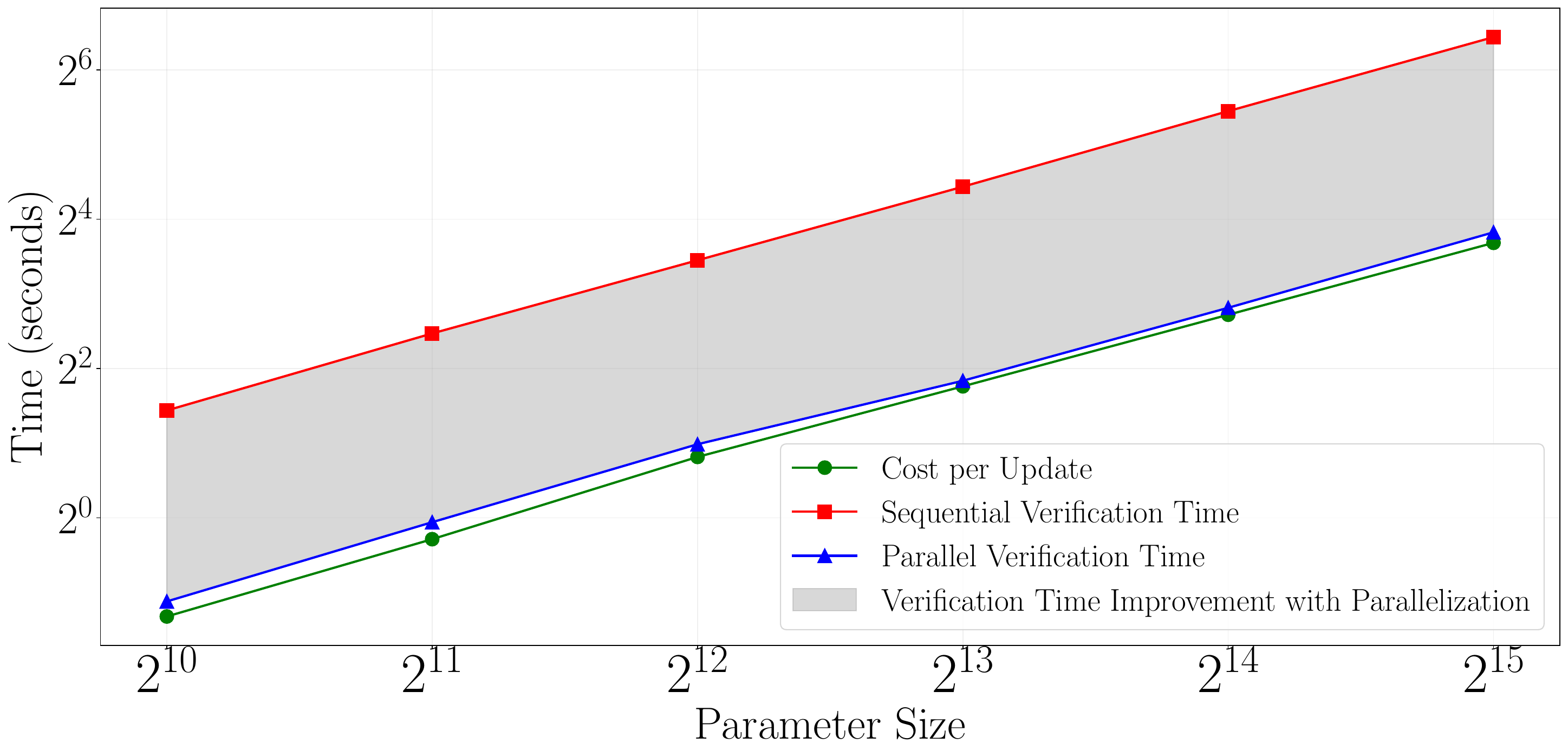}
    \caption{Off-Chain Computation Time for Contributors and Operators for $n\in \{2^{10}, 2^{11}, 2^{12}, 2^{13}, 2^{14}, 2^{15}\}$ and $k = 1$. 
    }
    \label{fig:offchain-cost}
\end{figure}


\pparagraph{Off-chain Aggregatable Contribution: Operator's Off-chain Costs}
We also benchmark the off-chain well-formedness verification.
In the batched contribution setting, 
contributors and operators can locally verify the $\pp$ with off-the-shelf optimizations. 
Notably, our approach enables parallel verification—a capability not available in fully on-chain ceremonies like \cite{boneh-pot-2024}'s scheme, which is constrained by the Ethereum Virtual Machine's (EVM) sequential execution model.
As shown in \Cref{fig:offchain-cost}, we achieve a ${\approx}6\times$ speedup in verification time using a $6$-core machine by parallelizing the addition operations in check 2.
These optimizations enable \sysname to efficiently process large volumes of contributions, demonstrating better scalability compared to existing approaches.


We plot the cost of off-chain contribution and verification for the operator and
contributor in~\Cref{fig:offchain-cost}.

\newcommand{\graycell}{\cellcolor{red!20}}


\begin{table*}[t!]
\centering
\caption{Estimated Gas (USD) Costs For Different Parameter Sizes. Red cells denote instances that cannot be executed in Ethereum as they exceed the block gas limit of 30M.
}
\resizebox{\textwidth}{!}{ 
\begin{threeparttable}
\begin{tabular}{ccccccc}
\toprule
$n$ & $2^{10}$ & $2^{11}$ & $2^{12}$ & $2^{13}$ & $2^{14}$ & $2^{15}$\\
\midrule
Total Cost of~\cite{boneh-pot-2024} & 10.9M \textbf{(806 USD)} & 23.1M \textbf{(1,707 USD)} & \graycell 51.9M \textbf{(3,836 USD)} & \graycell 127.3M \textbf{(9,410 USD)} & \graycell 349M \textbf{(25,797 USD)} & \graycell 851M \textbf{(62,904 USD)}\\
Our Optimistic Update Cost & 3.3M \textbf{(244 USD)} & 3.9M \textbf{(288 USD)} & 5.3M \textbf{(392 USD)} & 7.9M \textbf{(584 USD)} & 13.5M \textbf{(999 USD)} & 25.4M \textbf{(1,879 USD)}\\
Our Fraud Proof Verification Cost & 322,218 \textbf{(24 USD)} & 325,516 \textbf{(24 USD)} & 328,802 \textbf{(24 USD)} & 332,051 \textbf{(25 USD)} & 335,385 \textbf{(25 USD)} &  338,647 \textbf{(25 USD)}\\
\bottomrule
\end{tabular}
\begin{tablenotes}\footnotesize
\item[*] The USD values were calculated based on the Ethereum price of 3,080 USD and gas price of 24 gwei on Nov 14th, 2024.
\end{tablenotes}
\end{threeparttable}
} 
\label{tab:gas-cost}
\end{table*}

\section{Discussion}
\label{sec:discussion}

\pparagraph{Participation of Operators and Contributors}
The providers of the applications relying on a PoT string are inherently incentivized to act as operators, 
as the security of the PoT is fundamental to the security of the application.
To increase the confidence of users in the security of their application,
it is in the providers' interest to collect contributions from future users and provide inclusion proofs to them.
The providers can also set up incentive schemes, \eg, lottery schemes, to attract independent operators and contributors.
Hence, further incentive analysis can be an interesting future direction and orthogonal to our work.



\pparagraph{Handling Temporary Invalid Powers-of-Tau Parameters}
Our protocol allows potentially invalid parameters to remain on-chain for a short period of time in the absence of challenges. While honest users might use these ill-formed parameters for their application, this poses minimal risk for two reasons. First, the local verification of parameter well-formedness is computationally inexpensive, allowing users to efficiently validate parameters before use. Second, we could enhance security by adopting an approach similar to optimistic rollups, implementing a challenge period during which participants can submit fraud proofs to identify and reject invalid parameters. 
\pparagraph{Concurrency and Front-running} When multiple operators simultaneously attempt to update $\pp$ on $\contract$,
only one update will be accepted, every update else becomes useless due to the nature of sequential updates.
Nevertheless, one can deploy an orthogonal mechanism to resolve the concurrency issue.
For example, operators can register for the contract to distribute the update time slot.
A security concern is sybil attacks, where the attackers create many operators to occupy all time slots for a DoS attack. 
Still, many solutions are available, including:
    \emph{(i)} asking the operator to put collaterals on the contract;
    \emph{(ii)} selecting the operators using public randomness;
    \emph{(iii)} prioritizing operators with real-world identities; and many more.
We leave further study on concurrency control as a future work.

\pparagraph{Disincentive of Submitting Ill-formed $\pp$}
Attackers might submit ill-formed $\pp$ to thwart the ceremony. 
However, such attacks are economically irrational because anyone can defeat them 
using $10 \sim 78$-fold less financial cost, as shown in \Cref{tab:gas-cost}.
In the face of irrational attacks, our fraud-proof cost is also cheap enough (${\approx}$US\$ 25) for the providers and future users to disprove without financial concerns.

\section{Related Works}
\label{sec:related-work}


\begin{table*}[t!]
\centering
\caption{Comparison with Related Works for PoT Setup}
\label{tb:rel-work}
\resizebox{\textwidth}{!}{ 
\begin{threeparttable}
\begin{tabular}{lcccc}
  \toprule
  & Censorship Resistant & Fault Tolerance & Communication Model & On-Chain Cost       \\
  \midrule
  MPC-based PoT~\cite{bowe-pot-2017,ben-sasson-pot-2015}        & No               & ${m-1}$           & Round-Robin + Broadcast       & N/A                                      \\
  PoT in Async~\cite{das-asynchrony-pot-2022}  & No               & $m/3$           & Asynchronous      & N/A                                      \\
  PoT to the People~\cite{boneh-pot-2024} & \textbf{Yes}              & ${m-1}$           & Round-Robin + Blockchain      & $O(km)$ Pairing Ops. + $O(m (n+k))$ ECMult                       \\
  \sysname          & \textbf{Yes}              & ${m-1}$           & Round-Robin + Blockchain      & \textbf{$O(c)$ Pairing Ops. + $O(c (n+k)$ Calldata/Hashing} \\
  \bottomrule
\end{tabular}
\begin{tablenotes}\footnotesize
\item[*] $m$ is the number of contributors. $n$ and $k$ are the total degrees of the powers-of-tau string in $\gpone$ and $\gptwo$, respectively.
\item[*] $m-1$ fault tolerance is necessary for the more-the-merrier property.
\item[*] $m$ values in \cite{bowe-pot-2017, ben-sasson-pot-2015, das-asynchrony-pot-2022} are fixed while $m$ in \cite{boneh-pot-2024} and our protocol can keep growing as more users contribute to the setup. 
\item[*] $c$ is the number of (trustless) operators in our protocol and $c \ll m$.
\end{tablenotes}
\end{threeparttable}
} 
\end{table*}

\pparagraph{Round-Robin-Based Setup}
Ben-Sasson~\etal~\cite{ben-sasson-pot-2015} introduced one of the first setup ceremony protocols that can tolerate $m-1$ malicious parties (thus is a ``more-the-merrier'' protocol) with the use of multiplicative shares.
Following a round-robin, each party sequentially introduces their private randomness into the $\pp$ as a multiplicative share and broadcasts the updated $\pp$. The protocol concludes when all parties have made their contributions. As long as at least one honest participant has updated the parameters, the final output is secure, thus ensuring that the trapdoor $\tau$ remains hidden from the adversary.
As the protocol is executed in a round-by-round manner, only one party can act on the $\pp$ and broadcast the message to other parties via a broadcast channel like a blockchain.
However, this protocol limits the number of participants to a pre-selected set as its security proof requires a pre-commitment phase.

Bowe~\etal~\cite{bowe-pot-2017} avoids the pre-commitment phase by using a random beacon model, which is an abstraction of delayed randomness such as the hash value from a hash-chain or blockchain.
Skipping the pre-commitment phase enables dynamic participation for new contributors and thus perpetual powers-of-tau ceremony, where the ceremony is ever-going.
Kohlweiss~\etal~\cite{asiacrypt/KohlweissMSV21} optimizes Boew~\etal's protocol by eliminating the use of a random beacon model and proves its security in the algebraic group model.


\pparagraph{Decentralized Setup}
The above works cannot prevent censorship because these {\emph{centralized}} parties can censor contribution to the $\pp$, allowing an adversary to have full knowledge of the randomness (\ie, $r$s)
and thus learn $\tau$ from $\pp$.
Nikolaenko~\etal~\cite{boneh-pot-2024} proposes running a decentralized ceremony on a smart contract on Ethereum. The contract stores $\pp$ and allows anyone to join the ceremony as the blockchain is permissionless. Yet, the gas cost per update remains high, which results in a high monetary cost that discourages participation and limits the total degree of $\pp$.

\pparagraph{Asynchronous PoT Setup} 
Das~\etal~\cite{das-asynchrony-pot-2022} revisited the PoT setup within the context of an asynchronous network model. Their protocol demonstrates greater efficiency in terms of (off-chain) computation and communication costs compared to previous methods.
However, due to the inherent limitation of the asynchronous network setting, it can only tolerate up to $m/3$ Byzantine parties. 
Thus, this protocol does not have the more-the-merrier property and requires the users of the final $\pp$ to have trust in the majority of the contributors.

\pparagraph{Optimistic Paradigm and Fraud Proofs}
Optimistic rollups~\cite{ethereum-optimistic-rollups} is a layer-2 solution for Ethereum to scale the transaction throughput.
In a nutshell, the (level-1) blockchain accepts the off-chain transaction result without verifying it, 
unless a valid fraud proof is received.
Arbitrum~\cite{uss/KalodnerGCWF18} is a concrete scheme of optimistic rollups. 
It proposes a fraud-proof mechanism that a challenger can run a bisection game with the transaction submitter 
to pinpoint and prove the invalid operation in the transaction.
However, it incurs expensive logarithmic rounds of on-chain interactions.



\pparagraph{Rogue-Key Attack Mitigation}
BLS signature scheme~\cite{asiacrypt/BonehLS01} uses bilinear pairing for short signatures and aggregation of signatures. However, the rogue-key attack is a well-known issue of aggregation over BLS signatures. 
Ristenpart~\etal~\cite{eurocrypt/RistenpartY07} proposes proof-of-possession (POP) to protect BLS signatures against rogue-key attacks. It requires every secret key holder to register their keys by signing on the hash value of its public key to show its knowledge of its secret key.
However, it is not succinct enough for our case as we discussed in \Cref{sec:pro-key-agg}.
Boneh~\etal~\cite{asiacrypt/BonehDN18} proposes another way to defend against rogue-key attacks by combining public keys and signatures using random coefficients. These coefficients are derived via hashing all the involved public keys.
It has two drawbacks: \emph{(i)} a verifier (\eg, the contract or a contributor) needs access to all public keys and 
\emph{(ii)} it requires a pre-selected set of public keys 
which restricts the participation of late-join contributors. 
A generic approach to sample random coefficients is Fiat-Shamir heuristic,
where the verifier hashes the entire transcript to output random values.
In the context of a smart contract, it means $\contract$ needs to hash all the prior $\pp$s, $\vk$s, and $\sigma$s,
which is too expensive because $\contract$'s computation cost would become $O(t^{2})$ for $t$ on-chain updates, while our approach's is $O(t)$.







\section{Conclusion}

We have presented an efficient and fully decentralized PoT setup ceremony scheme \sysname,
which significantly reduces the participation cost and can support larger PoT strings for broader applications. Our key contributions include 
\emph{(i)} introducing an efficient fraud-proof mechanism to reduce the on-chain cost for each update; 
\emph{(ii)} proposing a proof aggregation protocol to enable batching the contribution; 
\emph{(iii)} defining the security notion of aggregated contribution and rigorously proving our \sysname's security under AGM;
and \emph{(iv)} evaluating \sysname's concrete cost and performance.




\section*{Disclaimer}
Case studies, comparisons, statistics, research and recommendations are provided ``AS IS'' and intended for informational purposes only and should not be relied upon for operational, marketing, legal, technical, tax, financial or other advice.  Visa Inc. neither makes any warranty or representation as to the completeness or accuracy of the information within this document, nor assumes any liability or responsibility that may result from reliance on such information.  The Information contained herein is not intended as investment or legal advice, and readers are encouraged to seek the advice of a competent professional where such advice is required.

These materials and best practice recommendations are provided for informational purposes only and should not be relied upon for marketing, legal, regulatory or other advice. Recommended marketing materials should be independently evaluated in light of your specific business needs and any applicable laws and regulations. Visa is not responsible for your use of the marketing materials, best practice recommendations, or other information, including errors of any kind, contained in this document.
\bibliographystyle{IEEEtranS}
\bibliography{reference}

\appendices
\section{Extended Discussion}

\pparagraph{Extension to Centralized Setting}\label{sec:ext-bucket}
Our technique may be of independent interest in centralized setups~\cite{bowe-pot-2017},
where a centralized operator handles all the updates and publishes the entire transcript for verification.
In this setting, 
a contributor must inspect all $\checkone$ proofs of the updates after its submission to confirm the inclusion,
leading to $O(M)$ computation for a ceremony with $M$ contributors ($M$ is the total number of contributors, while $m$ is that in a batched update in our setting). 
By applying our aggregation proof, the verification cost for a contributor is reduced to $O(\sqrt{M})$ using a bucketization trick.

More specifically, the operator divides $M$ contributions into $\sqrt{M}$ consecutive chunks,
each of which is considered as a batched update with $\sqrt{M}$ contributions.
In each batch, the $\checkone$ proofs are aggregated using our technique.
Then, the operator publishes all the aggregated proofs and the $(\pk, \pos)$s of all updates.
Now, the contributor only needs to inspect $\sqrt{M}$ $\pos$s in its batch (acting like $\user$) and runs $\sqrt{M}$ progressive key verification on the batches proofs (acting like $\contract$),
resulting in $O(\sqrt{M})$ computation.



\section{Overview of Previous PoT}
\subsection{Nikolaenko et at~\cite{boneh-pot-2024}'s $\checktwo$}\label{sec:beneh-well-check}
The smart contract only accepts an update if $\pp$ passes the following well-formedness check:
{\scriptsize
\begin{align*}
&e\left(\sum_{i=1}^{n} \rho_1^{i-1} P_i, \ \gentwo + \sum_{\ell=1}^{k-1} \rho_2^{\ell} Q_{\ell}\right) = e\left(\genone + \sum_{i=2}^{n-1} \rho_1^{i} P_i, \ \sum_{\ell=3}^{k} \rho_2^{\ell-1} Q_{\ell}\right)
\end{align*}
}

To overcome the issue that Ethereum's smart contract does not support ECMult over $\gptwo$, 
\cite{boneh-pot-2024} converts the computation into a $4k - 6$ pairing checking for $k > 1$:
{
\begin{align*}
    &\text{For } R = \sum_{i=0}^{n-2} \rho^i \cdot P_{i+1,j} : \\
    &\quad\quad e(\genone + \rho R, \ Q_{1,j}) = e(R + \rho^{n-1} P_{n,j}, \ \gentwo) \\
    &\text{For } t = 2 \ldots k-1 :\\
    &\quad e(P_{k-t,j}, Q_t) = e(P_{k,j}, \gentwo) \land e(\genone, Q_k) = e(P_{k,j}, \gentwo)
\end{align*}
}


\subsection{On Aggregation of Nikolaenko et at~\cite{boneh-pot-2024}'s $\checkone$}\label{sec:agg-schnorr}
Recall that \cite{boneh-pot-2024} proposes the following protocol for $\checkone$.
With $\pp$ and a random value $r$, the prover update $\pp$ to $\pp'$ with $r$ and compute the proof $\pi$ as following.
\begin{align*}
z &\sample \Zp^*;
h \leftarrow H(P_1' \mid\mid P_1 \mid\mid z \cdot P_1) \\ 
\pi &\leftarrow (z \cdot P_1, z + h \cdot r) \in \mathbb{G}_1 \times \mathbb{Z}_p
\end{align*}
Then, a verifier with accesses to $\pp, \pp', \pi$ outputs
\begin{align*}
  \pi_2 \cdot P_1 \isequal \pi_1 + H(P_1' \mid\mid P_1 \mid\mid \pi_1) \cdot P_1
\end{align*}

A natural attempt to extend it for aggregation is to multiply the new random value $r'$ to $\pi_1$ and 
obtain $r' z + r' r h $. While $r' r$ in the second term is the knowledge we want to prove, $h$ in the term is not the proper challenge value $H(r' P'_1 || r' P_1 || z' P_1')$.





\section{Missing Proofs}
\label{sec:missing-proofs}
\subsection{Proof for Theorem~\ref{thm:incl-snd}}
\label{subsec:prf-incl-snd}
\begin{proof}
  At a high-level, we prove that if there is an adversary $\advinclude$ that can win $\inclusionsoundgame$ with non-negligible probability, we can use $\advinclude$ to construct an algorithm $\mathbf{B}$ to break the $q$-SDH assumption.

  $\mathbf{B}$ will run a $\inclusionsoundgame$ with $\adv$ as a simulator.
  If $\advinclude$ produces a valid output,
  $\mathbf{B}$ can extract the coefficient $\advinclude$ used over the group elements,
  which allows $\mathbf{B}$ to recover the value to break $q$-SDH assumption.


  Now, we examine what group elements $\advinclude$ can obtain from function invocation during the game:
  \begin{compactitem}
    \item $\funonupdate$: $\contract$ provides
          \begin{align*}
            {\stcontract}_{i} = (\vk_{i}, {\presig}_{i}, {\cursig}_{i}) \in \gpone \times \gptwo \times \gptwo
          \end{align*}
          where might be contributed by $\advinclude$ or the simulator.
          To the benefits of $\advinclude$, we assume $\advinclude$ can access
          \begin{align*}
            \pp_{i} = (\tau_{i} \genone, \ldots) \in \gpone^{n} \times \gptwo^{k}
          \end{align*}
          for $i \in [1..\mu]$ up to the current round $\mu$
          and
          the valid
          \begin{align*}
          {\stkey}_{i} = ((\pk_{1, j}, \pos_{1, j}) | j \in [1..m_i]) \in (\gpone \times \gptwo)^{m_{i}}
          \end{align*}
          for $i \in [1..\mu]$ and $m_t$ is the (claimed) number of contributors for
          the $i$-th subbmision.

    \item $\funoffupdate$: the simulator receives $\presig$, $\cursig$ and $\pp$
    from $\advinclude$ and return
          to $\advinclude$
          \begin{align*}
    \pk \in \Zp, \pos \in \gptwo \presig', \in \gptwo, \cursig' \in \gptwo, \pp' \in \gpone^{n} \times \gptwo^{k}
          \end{align*}

    \item $\funcheckinclude$: it is a local computation and have no interaction
          with the challenger, so it gives no additional group elements to
          $\advinclude$.

  \end{compactitem}

  At the end of the game, the blockchain will have $\ell$ submissions and $\advinclude$ outputs
  \begin{align*}
    \stser^{*} = (\pp^{*}, \presig^{*}, \cursig^{*}, \cdot, \curvk^{*}, \cdot, \stkey^{*})
  \end{align*}
  that passes $\checkinclusion(\pk, \stser^{*}, \txser)$
  and $w \in \Zp$ such that $\pp_{\ell} = (w \genone, \ldots, w^{n} \gentwo ; \quad w\gentwo, \ldots, w^{k} \gentwo)$, where $\pp_{\ell}$ is the latest valid PoT string on $\contract$.
  For the sake of simplicity, we assume all submisisons on the blockchain are valid and the $\pp'$ are well-formed.

  Additionally, we assume that $\advinclude$ uploads to the $\contract$ the update containing the simulator's $\presig$ and $\cursig$ $t$-th update via $\funonupdate$. (If $\advinclude$ has multiple uploads containing $\presig$ and $\cursig$, we consider the first time $\advinclude$ submit as $t$.)

  We now introduce some notation in this proof. We use $\log : \gpone \rightarrow \Zp$ to denote the scalar $\log G \in \Zp$ such that $(\log G) \cdot 1_{\gpone} = G$. To slightly abuse the notation, we denote $\log$ for $\gptwo$ similarly. We remains use the addition notation for $\log(\cdot)$. In other words, if $G = a \genone + b \genone$, we write $\log G = a + b$.
  Additionally, we denote $\log \pk = \sk$.
  Note that $\log$ is merely a notation and does not imply an efficient algorithm that can compute the logarithm out of a group element.

  As $\advinclude$ is an algebraic adversary, whenever $\advinclude$ outputs an group element $G$, it also outputs an array of coefficient $\vec{c_{G}}$ such that $G = \sum_{i} c_{G}[U_{i}] \cdot U_{i}$, where $\{U_{i}\}_{i}$ are all the group elements $\advinclude$ previously obtained. Syntactically speaking, the $U_{i}$ in $c_{G}[U_{i}]$ is treated as a symbol instead of a value.

  For now, we also assume $t = \ell$, \ie, the $t$-th update is the last update of the entire challenge.

  \begin{lemma}\label{lem:coeff-sk}
    Assume $t = \ell$.
    For any group element $G$ that $\advinclude$ outputs after the interaction
    with the simulator via $\funcheckinclude$, their logarithmic can be written
    as a polynomial of $\sk$:
    \begin{align*}
      \log G = A_{G} \cdot \sk + B_{G} \in \Zp
    \end{align*}
    where $A_{G}$ and $B_{G}$ are terms that does not contain $\sk$.
    If $G \in \gpone$, we have
    \begin{align*}
      A_{G} = \coeff{G}{\pk}
    \end{align*}
    If $G \in \gptwo$, we have
    \begin{align*}
      A_{G} = \coeff{G}{\cursig'} \cdot r \cdot \xi + \coeff{G}{\pos} \cdot \log H(\pk)
    \end{align*}
    where $\xi = \log P_1$ and $\pp = (\ldots; \quad Q_{1}, \ldots)$ is the first $\gpone$ elements in the PoT string that $\advinclude$ sends to the simulator in $\funoffupdate$.
  \end{lemma}

  \begin{proof}
    By inspecting \Cref{alg:off-update}, a contributor only uses its $\sk$ to produce $\pk \in \gpone$ and $\pop, \cursig' \in \gptwo$.
  \end{proof}

As $\advinclude$ wins the game, the above variables satisfy:

{\small 
\begin{align*}
e(\rho_{t, 1} \cdot \vk_{t-1} + \rho_{t, 2} \cdot \vk^{*}, Q_{1}^{*}) = e(\genone, \rho_{t, 1} \cdot \presig^{*} + \rho_{t, 2} \cdot \cursig^{*})
\end{align*}
}

where $\log Q^{*}_{1} = \tau_{t} = w$ and $\rho_{t, j} = H(\vk_{t-1}||\vk^{*}||\presig^{*}|| \cursig^{*} || Q_{1}^{*} || j)$ for $j \in \{1, 2\}$.

Before analysis the the above constraint, we introduce a lemma related to $\rho$.
\begin{lemma}\label{lem:hash-zippel}
  Suppose
  \begin{align*}
   0 = \rho_{j, 1} D_{j, 1} + \rho_{j, 1} D_{j,2}
  \end{align*}
  for some $j \in \mathbb{N}$
  and
   $D_{k, 1} = \rho_{k-1, 1} D_{k-1, 1} + \rho_{k-1, 1} D_{k-1,2} \in \Zp$
  for all $k \in [i+1..j]$ and some $i \in \mathbb{N}$.

  If changing of the value of $D_{k, *}$ for any $k \in [i..j]$ implies a change of the value of
  $\rho_{k, *}$,
  then we have $D_{k} = 0$ for all $k \in [i..j]$ with probabilty $\geq 1 - \numhashquery/p$.
\end{lemma}


\begin{proof}
  We expand $D_{k, 1}$ into
  $ (\rho_{k+1, 1} \ldots\rho_{i+1,1}(\rho_{i,1} D_{i, 1} + \rho_{i,2} D_{i, 2}) +\ldots)+ \rho_{k+1, 2} D_{k+1, 2})$
  for any $k \in [i..j]$.
  $(\rho_{k, 1} D_{k, 1} + \rho_{k, 2} D_{k, 2})$ can be considered as a multivariate polynomial of $\rho_{k, 1}$ and $\rho_{k, 2}$ with total degree of $2$.
   Since the value of $\rho_{k, 1}$ and $\rho_{k, 2}$ are sampled after $D_{k, 1}$ and $D_{k, 2}$, and thus the polynomial is evaluated at random points.
  Then by the Schwartz-Zippel lemma, $D_{k, 1} = D_{k, 2} = 0$ with $1 - 2/p$ probability.
  Given that $D_{k, 1} = 0$, $D_{k+1, 1} = D_{k+1, 1} = 0$ with $1 - 2/p$ probability for the same reason.
  By induction, we know $D_{k, 1} = D_{k, 1} = 0$ with probability
  $\geq (1 - 2/p)^{k-i+1}$ for any $k \in [i..j]$ \emph{if} we do not apply the Fait-Shamir Heuristic.
  However, in our protocol, $\rho_{k, *}$ are sampled through a random oracle, so the above step takes
  $2(j-i+1)$ queries to the random oracle.
  Notice that an adversary can also change the values of $D_{k, *}$ to resample $\rho_{k, *}$, and each sampling takes $2$ hash queries.
  Hence, the probability of violating all $D_{k, *}$ being zeros is further increased to $\leq 1 - (1-2/p)^{\numhashquery/2} \leq \numhashquery/p$.
\end{proof}

\begin{remark}
Notice that
$\rho_{j, 1} D_{j, 1} + \rho_{j, 1} D_{j,2}$
can be rewrite into another form of
$\sum_{k=i}^j (\prod_{\nu=i+1}^{j} \rho_{\nu,1}) (\rho_{k, 1} D_{k, 1} + \rho_{k, 2} D_{k, 2})$.
\end{remark}

\begin{remark}
  The quota of random oracle queries, \ie, $\numhashquery$, are ``shared'' in the entire execution of $\advinclude$.
  Hence, if in our proof, we apply Lemma~\ref{lem:hash-zippel} multiple times, the overall failture probability all of algorithm is still $\leq \frac{\numhashquery}{p}$, instead of a multiple of it.
\end{remark}

Now, we proceed our analysis on the contraint.
The equality implies

{\small 
\begin{align*}
  w \cdot (\rho_{t, 1} \log \vk_{t-1} + \rho_{t, 2} \log \vk^{*})
  = \rho_{t, 1} \log \presig^{*} + \rho_{t, 1} \log \cursig^{*} \\
  0 = \rho_{t, 1} (w \cdot \log \vk_{t-1} - \log \presig^{*}) + \rho_{t, 2} (w \cdot \log \vk_{t-2} - \log \cursig^{*})
\end{align*}
}



Now we apply Lemma~\ref{lem:hash-zippel} on it.
$w = \tau_{t} = \log Q^{*}_{1}, \vk_{t-1}, \presig^{*}$ are used as part of the query to the random oracle $\hashzp$ to produce $\rho_{t, *}$.
Changing any of them will change $\rho_{t, *}$.
By Lemma~\ref{lem:hash-zippel} and
by setting $D_{t, 1} = w \cdot \log \vk_{t-1} - \log \presig^{*}$,
we have $w \cdot \log \vk^{*} - \log \cursig^{*} = 0$ with probability $1 - \numhashquery/p$.

We rewrite $w \cdot \log \vk^{*} - \log \cursig^{*} = A_{t} \sk + B_{t}$ where
\begin{align*}
A_{t} = w \cdot \coeff{\vk^{*}}{\pk} + \coeff{\cursig^{*}}{\cursig'} \cdot r \cdot \xi + \coeff{\cursig^{*}}{\pos} \cdot \log H(\pk)
\end{align*}


We denote the following events:
\begin{compactitem}
  \item $\event{A}$ the event $A_{t} \neq 0$
  \item $\event{\mathsf{cur}}$ the event $\coeff{\cursig^{*}}{\cursig'} \neq 0$
  \item $\event{\pk}$ the event $\coeff{\vk^{*}}{\pk} \neq 0$
\end{compactitem}

In the following lemmas, we will construct multiple algorithm who can extract a solution of $(n,k)$-SDH instance from $\advinclude$ with overwhelming probability,
each conditioning on the occurrence of above events.
These conditions altogether will cover all the possible events.
Then, $\mathcal{B}$ can randomly select one of these sub-algorithms to extract the solution from $\advinclude$.

\begin{lemma}
Given that $t = \ell$, there exist PPT algorithms $\calg, \dalg, \ealg$ such that

{\small 
\noindent$\pr{\gamesdh^{\calg} = 1} = \pr{\inclusionsoundgame = 1 | \event{A} = 1}$, 
    $\pr{\gamesdh^{\dalg} = 1} \geq \frac{p-1}{p} \pr{\inclusionsoundgame = 1 | \event{A} = 0 \wedge \event{\mathsf{cur}} = 1}$,
    $\pr{\gamesdh^{\ealg} = 1} \geq \\
    \frac{p-2}{p} \pr{\inclusionsoundgame = 1 | \event{A} = 0 \wedge \event{\mathsf{cur}} = 0 \wedge \event{\pk} = 1}$.
}

\end{lemma}

\begin{proof}
  Let $(\alpha_1 = z \genone, \ldots, \beta_k = z^{k} \gentwo)$ be the $(n, k)$-SDH instance.

  (All of the following adversaries programs $\hashzp(\cdot)$ normally, \ie, for $\hashzp(x_{i})$, they outputs $r_{i} \sample \Zp$ if $x_{i}$ is a first-time query; otherwise, they responses with the previous $\hashzp(x_{i})$.)

  \textbf{Adversary} $\calg(\alpha_1 = z \genone, \ldots, \beta_k = z^{k} \gentwo)$:
  At the high-level, $\calg$ embeds the SDH instance into the $\sk$ (and thus $\pk$).

  It sets $\pk = \alpha_{1} = z \cdot \genone$.
  For the random oracle $\hashgp(\cdot)$,
  When someone queries $H(x_{1})$, it samples $r_{i} \gentwo \in \Zp$ and sets $H(x_{1}) = r_{1} \gentwo$ if $H(x_{1}) = \bot$ (\ie, $x_{1}$ is not queried before); otherwise, it returns the $H(x_{1})$ from its records.

  $\calg$ invokes $\funonupdate$ periodically following the real-world distribution of invocation and performs providing valid inputs like honest servers to simulate view.
  When $\advinclude$ queries $\calg$ for $\funoffupdate$, $\calg$ provides its input as an honest contributor except $\calg$ needs special steps to handle the signing for $\pk$ (recall that $\calg$ does not know $\sk = z$):
  \textit{i)} To compute $\sigma'$, $\calg$ samples $r$ and retrieves from $\advinclude$ the coefficients to construct $Q_{1}$, \ie, $ \chi = \log Q_{1}$, it then computes $\sigma' = r\cdot \cursig + r \log Q_{1} \cdot \pk$.
  \textit{ii)}
  To compute $\pos$, $\calg$ retrieves (or samples) $r_{\pk}$, which is the random coefficient programmed for $H(\pk)$. It then computes $\pos = r_{\pk} \cdot \pk$.
  Notice that $\calg$ perfectly simulates $\advinclude$'s view.

  Give that $\event{A} = 1$, \ie, $A_{t} \neq 0$,
  $\calg$ extracts $z$ from $\advinclude$'s output by computing $z = \sk = B_{t}/A_{t}$.
  With $z$, $\calg$ can provides $(c + z, \frac{1}{c+z}\genone)$ for any $c \in \Zp\backslash \{-z\}$

  \textbf{Adversary} $\dalg(\alpha_1 = z \genone, \ldots, \beta_k = z^{k} \gentwo)$:
  At a high-level, $\dalg$ embeds the $(n,k)$-SDH instance $r$ (the random value multipled into $\pp_{t}$ by $\dalg$ as an contributor).


  $\dalg$ programs $\hashgp(\cdot)$ like $\calg$ does and simulates the view about $\funonupdate$ as the honest case.
  For $\funoffupdate$, $\dalg$ interacts with $\advinclude$ like an honest contributor except when computing for $\pp'$ and $\cursig'$.
  Notice that $\dalg$ uses any information about the $(n,k)$-SDH instance to generate the view for $\advinclude$ thus far.
  All the scalar values $\advinclude$ used to construct $\pp$ are known by $\dalg$ because the invoked scalar values are either returned by $\advinclude$ as an algebraic adversary or sampled by the $\dalg$ itself.
  In other words, $\dalg$ knows $\xi$.
  (If $\dalg$ cannot extract the scalar from $\pp$ because of $\pp$'s ill-formedness, $\dalg$ aborts like an honest user.)
  $\dalg$ can compute
  $\pp' = (\xi \cdot \alpha_{1}, \ldots; \quad \ldots \xi^{k} \cdot \beta_{k})$.
  We can easily see that the logarithm of them are
  $(\xi \cdot z, \ldots; \quad \ldots, (\xi \cdot z)^{k})$.
  Similarly, $\dalg$ extracts $\log \cursig$ and then computes $\cursig' = \log \cursig \cdot \beta_{1} + \sk \cdot \tau_{t-1} \cdot \xi \cdot \alpha_{1}$. $\dalg$ perfectly simulate $\advinclude$'s view.

  Now, we analyze how $\dalg$ extract $z$ from $\advinclude$'s output.
  Given that $A_{t} = 0$ and $\coeff{\cursig^{*}}{\cursig'} = 1$, $\dalg$ can extract

  \begin{align*}
   z = -\frac{w \cdot \coeff{\vk^{*}}{\pk} + \coeff{\cursig^{*}}{\pos} \cdot \log \hashgp(\pk)}{\coeff{\cursig^{*}}{\cursig'} \cdot \tau_{t-1} \cdot \xi }
  \end{align*}

  If $\xi = 0$,
  $\dalg$ would have seen $\presig = 0_{\gptwo}$ during $\funoffupdate$ and thus would have aborted like an honest contributor.

  \textbf{Adversary} $\ealg(\alpha_1 = z \genone, \ldots, \beta_k = z^{k} \gentwo)$:
  At a high-level, $\ealg$ embeds the $(n,k)$-SDH instance into the random oracle $\hashgp$'s responses.

  When $\advinclude$ queries $\hashgp(x_{i})$ for $x_{i} \in \{0, 1\}^{*}$, $\ealg$ samples $a_{i}, b_{i} \sample \Zp \times \Zp$ and sets $\hashgp(x_{i}) = a_{i}\cdot \beta_{1} + b_{i}\cdot \gentwo$ if $\hashgp(x_{1}) = \bot$; otherwise, it returns the previously set $\hashgp(x_{i})$.
  Since $a_{i}, b_{i}$ perfectly blind $z$, $\ealg$ perfectly simulates the view about querying $\hashgp(\cdot)$.

  $\calg$ invokes $\funonupdate$ periodically as honest servers for simulation.
  When $\advinclude$ queries $\ealg$ for $\funoffupdate$,
  $\ealg$ provides its input as if an honest contributor.

  Given that $A_{t} = 0$ and $\coeff{\cursig^{*}}{\cursig'} = 0$,
  we have
  \begin{align*}
    w \cdot \coeff{\vk^{*}}{\pk} + \coeff{\cursig^{*}}{\pos} \cdot \log \hashgp(\pk) = 0
  \end{align*}

  Recall that $\log \hashgp(\pk) = a_{\pk} z + b_{\pk}$, where $a_{\pk}$ and $b_{\pk}$ are the sampled value during the first query of $\hashgp(\pk)$.
  We can rewrite the above equation as
  \begin{align*}
    z = (b_{\pk} - \frac{w \cdot \coeff{\vk^{*}}{\pk}}{\coeff{\cursig^{*}}{\pos}})/a_{\pk}
  \end{align*}

  As $\pr{a_{\pk} \neq 0} = \frac{p-1}{p}$ and as we are given that $\event{\mathsf{cur}} = 1$, \ie, $\coeff{\cursig^{*}}{\pos} \neq 0$, $\ealg$ can extract $z$ with probability $\frac{p-1}{p}$

\end{proof}

The following lemma states that as long as $\coeff{\vk^{*}}{\pk} = 0$, our algorithm can extract the solution for the $(n, k)$-SDH instance without knowing the value of $\tau_{t}$ or $w$.
We denote by $\event{\pk}$ the event that $\coeff{\vk^{*}}{\pk} \neq 0$.

\begin{lemma}\label{lem:incl-sound-nopk}
  Without assuming $t = \ell$,
  there exists a PPT algorithm $\falg$
  such that
\begin{align*}
  \pr{\mathbf{sdh}^{\falg} = 1} &\geq \frac{p-2}{2p} \pr{\inclusionsoundgame = 1 | \event{\pk} = 0}
\end{align*}
\end{lemma}

\begin{proof}
  \textbf{Adversary}
  $\falg(\alpha_1 = z \genone, \ldots, \beta_k = z^{k} \gentwo)$:
  We now analyze the case where $\coeff{\vk^{*}}{\pk} = 0$.

  Notice that $\checkinclusion(\pk, \stser, \txser) = 1$ (which is one of the winning condition) guarantees $\vk^{*} = \pk_{t, 1} + \ldots + \pk_{t, m_{t}}$ and $\pk$ is one of the $\{\pk_{t, i}\}_{i \in [1..m_{t}]}$,
  Thus, we can rewrite it as $\log \vk^{*} = \sk_{t, 1} + ... + \sk_{t, m_{t} - 1} + \sk$.
  Assume that $\falg$ can extract from $\advinclude$ all the coefficients $\sk_{t, 1}, \ldots, \sk_{t, m(t) - 1}$,
  $\falg$ will be able to compute $\sk$.

  As an $\gpone$ elements,
  we also know $\log \vk^{*} = \coeff{\vk^{*}}{\pk} \cdot \sk + B_{\vk^{*}} = B_{\vk^{*}}$.
  The later equality is due to  $\coeff{\vk^{*}}{\pk} = 0$.
  We have $\sk = B_{\vk^{*}} - \sum_{j=1}^{m_{t}-1} \sk_{t, j}$.
  To leverage this fact, $\falg$ can run algorithm $\faalg$ that embeds the $(n,k)$-SDH instance into $\pk$ like $\calg$ does and extract $z = \sk$.

  However, $\falg$ may not always be able to extract all $\{\sk_{t, j}\}_{j=[1..m_t]}$
  because $\advinclude$ may use $\pk$ to from some $\pk_{t, j}$.
  In this case, $\falg$ cannot extract $\sk_{t, j}$ if the $(n,k)$-SDH instance is embedded into $\pk_{t, j}$.

  To handle this situation, $\falg$ could instead make use of another algorithm $\fbalg$ that embeds the $(n, k)$-SDH instance into the random oracle like $\ealg$.
  $\fbalg$ programmed $\hashgp(\cdot)$ such that
  $\hashgp(\pk_{t, i}) = a_{\pk_{t, i}} z + b_{\pk_{t, i}}$ and $\hashgp(\pk) = a_{\pk} z + b_{\pk}$
  where $a_{\pk_{t, i}}, b_{\pk_{t, i}},  a_{\pk}, b_{\pk} \sample \Zp$.

  Observe that $\funcheckinclude$ would verify
  \begin{align*}
   e(\pk_{t, i}, \hashgp(\pk_{t, i})) \isequal e(\genone, \pos_{t, i})
  \end{align*}
  for all $i \in [1..m_{t}]$.
  It enforces $\sk_{t, i} \cdot \log \hashgp(\pk_{t, i}) = \log \pos_{t, i}$.
  As $\pos_{t, i}$ is a $\gptwo$ elements, we can write $\log \pos_{t, i} = A_{\pos_{t, i}} \sk + B_{\pos_{t, i}}$ where
  $A_{\pos_{t, i}} = \coeff{\pos_{t, i}}{\cursig'} \cdot r \cdot \xi + \coeff{\pos_{t, i}}{\pos} \cdot \log \hashgp(\pk)$

  Combining the above four equations, we have
  $0 = A_{t, i}  \cdot \sk + B_{t, i}$ for some terms $A_{t, i}, B_{t, i}$ independent of $\sk$
  and specially,
  $A_{t, i} = \coeff{\pos_{t, i}}{\pk} (a_{\pk_{t, i}} z + b_{\pk_{t, i}}) - (\coeff{\pos_{t, i}}{\cursig'} \cdot r \cdot \xi  + \coeff{\pos_{t, i}}{\pos} \cdot a_{\pk} z + b_{\pk})$

  Notice that we have $A_{t, i} \eqp 0$ for all $i \in [1..m_{t}]$; otherwise, $\faalg$, the other sub-algorithm of $\falg$ who embeds the problem instance into $\sk$, could compute $\sk = B_{t, i}/A_{t, i}$ for some $i$ such that $A_{t, i} \not\eqp 0$
  and extract the $z$ from $\sk$ (Here, we are implicitly adding another extraction step for $\faalg$).
  Now, $\fbalg$ could safely assume $A_{t, i} \eqp 0$ and compute
  \begin{align*}
    z \eqp \frac{\coeff{\pos_{t, i}}{\cursig'} \cdot r \cdot \xi
    - \coeff{\pos_{t, i}}{\pos} \cdot b_{\pk}
    - \coeff{\pk_{t, i}}{\pk} \cdot b_{\pk_{t, i}}}
    {\coeff{\pk_{t, i}}{\pk} \cdot a_{\pk_{t, i}} - \coeff{\pos_{t, i}}{\pos} \cdot a_{\pk}}
  \end{align*}

  Now, we analysis the success probability of $\falg, \faalg, \fbalg$.
  We denote by $\event{A_{t}}$ the event that $A_{\pos_{t, i}} \neq 0$ for all $i$ and
  by $\event{\pk_{t, *}}$ the event that $\coeff{\pk}{\pk_{t, i}} = 0$ for all $i$.
  We have

$        \pr{\gamesdh^{\faalg} = 1} = \\\pr{\inclusionsoundgame = 1 | \event{\pk} = 0 \wedge (\event{A_{t}} = 1 \vee \event{\pk_{t, *}} = 1)} $
and
$\pr{\gamesdh^{\fbalg} = 1}  \geq \\ \frac{p-2}{p} \pr{\inclusionsoundgame = 1 | \event{\pk} = 0 \wedge (\event{A_{t}} = 0 \wedge \event{\pk_{t, *}} = 0)}$.

  We define as $\falg$ as an algorithm who randomly executes either $\faalg$ or $\fbalg$, each with probability $1/2$.
  We have
    \begin{align*}
        \pr{\gamesdh^{\falg} = 1} & \geq \frac{p-2}{2p} \pr{\inclusionsoundgame = 1 | \event{\pk} = 0}
  \end{align*}
\end{proof}

Now, we forgo our assumption that $t = \ell$. It enables $\advinclude$ to submit more valid update on-chain interacting with the challenger, which may or may not contain the challenger's contribution.

\begin{lemma}\label{lem:snd-after-t}
  There exist PPT algorithms $\jalg$ such that
  \begin{align*}
    \pr{\gamesdh^{\jalg} = 1} &\geq \frac{2p - \numhashquery - 2}{6p} \pr{\inclusionsoundgame = 1 | \event{\pk} = 1}
  \end{align*}
\end{lemma}

\begin{proof}
  Now we explain how to extend $\dalg, \ealg, \falg$ to $\dalg', \ealg', \falg'$ to handle the case where $t < \ell$.
  The idea is to prove that $\dalg', \ealg', \falg'$ will still be able to extract the coefficient $\advinclude$ used to construct $\tau_{\ell}$ from $\tau_{t}$.
  Then we construct $\jalg$ from them.

  We start with showing that one of the final verification key $\vk_{\ell-1}$ must ``contains'' the simulator's $\pk$ with overwhelming probability.
  Speaking precisely,
  if we write $\log \vk_{\ell} = \coeff{\vk_{\ell}}{\pk} \sk + B_{\vk_{\ell}}$, then we have $\coeff{\vk_{\ell}}{\pk} \neq 0$ with overwhelming probability.
  Note that
  {\small 
  \begin{align*}
    \vk_{\ell-1}
    &= \rho_{\ell-1, 1} \vk_{\ell - 2} + \rho_{\ell-1, 1} {\curvk}_{\ell-1}
    \\&= \sum_{i=t-1}^{\ell-1} (\prod_{j=i+1}^{\ell-1} \rho_{j,1}) (\rho_{i, 1} \vk_{i} + \rho_{i, 2} {\curvk}_{i})
    \\ \log \vk_{\ell-1}
    &= \sum_{i=t-1}^{\ell-1} (\prod_{j=i+1}^{\ell-1} \rho_{j,1}) (\rho_{i, 1} \coeff{\vk_{i}}{\pk} + \rho_{i, 2} \coeff{{\curvk}_{i}}{\pk} ) \pk + \ldots
  \end{align*}
   }
  Notice that
  $\coeff{\vk_{\ell-1}}{\pk} = \sum_{i=t-1}^{\ell-1} (\prod_{j=i+1}^{\ell-1} \rho_{j,1}) (\rho_{i, 1} \coeff{\vk_{i}}{\pk} + \rho_{i, 2} \coeff{{\curvk}_{i}}{\pk} ) = 0$.
  If
  $\coeff{\vk_{\ell-1}}{\pk} = 0$,
  then by Lemma~\ref{lem:hash-zippel}, we have $\coeff{{\curvk}_{i}}{\pk} = 0$ for all $i \in [t-1..\ell-1]$ with overwhelming probability.
  However, it contradicts with
  the premise of this lemma that $\coeff{{\curvk}_{t}}{\pk} \neq 0$ (as $\event{\pk} = 1$).
  Hence, $\coeff{\vk_{\ell-1}}{\pk} \neq 0$ with probability $1- \numhashquery/p$.


  For easier analysis, we for now assume all the $t$-th to $\ell$-th submissions to $\funonupdate$ are submitted by $\advinclude$. It means $\advinclude$ obtain no additional group elements after $\advinclude$ interacting with
  $\advb$'s algorithms
  for $\funoffupdate$.

  Recall that when $\advinclude$ wins the game, it has to satisfy the following constraint:
  $e(\rho_{\ell, 1} \cdot \vk_{\ell-1} + \rho_{\ell, 2} \cdot \vk^{*}, Q_{1}^{*}) = e(\genone, \rho_{\ell, 1} \cdot {\presig}_{\ell} + \rho_{\ell, 2} \cdot {\cursig}_{\ell})$



  Again, by applying Lemma~\ref{lem:hash-zippel} and taking logarithm to the group elements, we have
  with $\geq 1 - \frac{\numhashquery}{p}$ probability,
  $0 = A_{>t} \sk + B_{>t}$ for some terms $A_{>t}, B_{>t}$ that does not contains $\sk$ and
  $  A_{>t} = w \cdot \coeff{\vk_{\ell-1}}{\pk} + \coeff{{\presig}_{\ell}}{\cursig'} \cdot r \cdot \xi + \coeff{{\presig}_{\ell}}{\pos} \cdot \log H(\pk)$

  
  Now, we can apply our prior strategy to construct $\calg', \dalg', \ealg'$.
  If $\ell > t$,
  $\calg'$ embeds the $(n,k)$-SDH instance into $\sk$ to extract $z$ if $A_{>t} = 0$.
  $\dalg'$ embeds the instance into $r$ to extract $z$ if $\coeff{{\presig}_{\ell}}{\cursig'} \neq 0$.
  $\ealg'$ embeds the instance into $\hashgp$'s response to extract $z$ if $\coeff{{\presig}_{\ell}}{\cursig'} \neq \coeff{{\presig}_{\ell}}{\pos} \neq 0$.
  If $\ell = t$, $\calg', \dalg', \ealg'$ behave exactly like $\calg, \dalg, \ealg$, respectively.



  Now, we discuss that the cases where some of the $t+1$-th to $\ell$-th submissions are not made $\advinclude$ but by the simulator.
  In these cases, $\dalg'$ or $\ealg'$ would not uses their $\sk$ to simulate the submission, and thus give no additional advantage to $\advinclude$ to preventing the simulator from extracting the solution by leveraging the above equations.
  %
  Finally, we construct $\jalg$ as an algorithm uniformly selecting one of $\calg', \dalg', \ealg'$ to execute, each with probability $1/3$.
  Thus, the worst-case success probability of $\jalg$ is $1/3$ times minimal of the success probability of $\calg', \dalg', \ealg'$ and the success probability in Lemma~\ref{lem:hash-zippel}, which is
  \begin{align*}
    \frac{2p - \numhashquery - 2}{6p} \pr{\inclusionsoundgame = 1 | \event{\pk} = 1}
  \end{align*}
\end{proof}

Finally, we construct $\advb$ as an algorithm who randomly select $\falg$ and $\jalg$, each with probability $1/2$.
By Lemma~\ref{lem:incl-sound-nopk} and Lemma~\ref{lem:snd-after-t}, we have
\begin{align*}
    \pr{\gamesdh^{\advb} = 1} &\geq \frac{2p - \numhashquery - 2}{12p} \pr{\inclusionsoundgame = 1}
\end{align*}

\end{proof}





\end{document}